\documentclass[11pt]{article}

\usepackage{amsthm,amsmath,amsfonts, amssymb}
\usepackage{mathtools}
\usepackage{fullpage}
\usepackage[hidelinks]{hyperref}
\usepackage{mathtools}
\usepackage{verbatim}
\usepackage{color}
\usepackage{enumitem}
\usepackage{caption}

\usepackage{tikz}
\usetikzlibrary{positioning,matrix,calc,arrows,shapes,fit,decorations,angles,quotes}
\usepackage{algorithm}
\usepackage[noend]{algpseudocode}

\usepackage{thmtools}
\usepackage{thm-restate}

\DeclarePairedDelimiter\ceil{\lceil}{\rceil}
\DeclarePairedDelimiter\floor{\lfloor}{\rfloor}

\newtheorem{lemma}{Lemma}
\newtheorem{theorem}{Theorem}
\newtheorem{corollary}{Corollary}
\newtheorem{definition}{Definition}

\newtheorem{proposition}{Proposition}

\newcommand{\eps}{\varepsilon}

\newcommand{\poly}{\mathrm{poly}}
\newcommand{\bm}[1]{#1}

\newcommand{\cB}{\mathcal{B}}

\newcommand{\cQ}{\mathcal{Q}}

\newcommand{\cD}{\mathcal{D}}

\newcommand{\F}{\mathbb{F}}
\newcommand{\rk}{\mathsf{rk}}

\begin{document}
\captionsetup{justification=justified}
\begin{titlepage}
\title{Static Data Structure Lower Bounds Imply Rigidity}
\author{Zeev Dvir\thanks{Department of Computer Science and Department of Mathematics,
Princeton University.
Email: \texttt{zeev.dvir@gmail.com}. Research supported by NSF CAREER award DMS-1451191 and NSF grant CCF-1523816.}
\and 
Alexander Golovnev\thanks{Harvard University. Email: \texttt{alexgolovnev@gmail.com}. Research supported by a Rabin Postdoctoral Fellowship.}
\and
Omri Weinstein\thanks{Columbia University. Email: \texttt{omri@cs.columbia.edu}. 
Research supported by NSF CAREER award CCF-1844887.}
}
\date{}
\maketitle
\thispagestyle{empty}

\begin{abstract}
We show that static data structure lower bounds in the group (linear) model imply semi-explicit  lower bounds on matrix rigidity.  In particular, we prove that an explicit lower bound of $t \geq \omega(\log^2 n)$ on the cell-probe complexity of %
linear data structures in the group model, 
even against arbitrarily small linear space $(s= (1+\eps)n)$, would already imply a semi-explicit ($\bf P^{NP}\rm$) construction of %
rigid matrices with significantly better parameters than the current state of art (Alon, Panigrahy and Yekhanin, 2009). Our results 
further assert that polynomial 
($t\geq n^{\delta}$) data structure lower bounds against near-optimal space, would imply super-linear circuit lower bounds for log-depth linear circuits (a four-decade open question). 
In the succinct space regime $(s=n+o(n))$, we show that any improvement on current 
cell-probe lower bounds in the linear model would also imply new rigidity bounds. Our results rely on a new 
connection between the ``inner" and ``outer" dimensions of a matrix (Paturi and Pudl{\'a}k, 2006), and on a  
new reduction from worst-case to average-case rigidity, 
which is of independent interest. 

\end{abstract}

\end{titlepage}
\section{Introduction}

Proving %
lower bounds on the operational time of data structures has been a long and active 
research endeavor for several decades. %
In the static setting, the goal is to \emph{preprocess} a database of $n$ 
elements into minimum space $s$ ($\geq n$), so that queries $q \in \cQ$ on the input database 
can be answered quickly, in query time $t$ (where the typical and realistic setting is $|\cQ|=\poly(n)$). %
The two na\"ive solutions to any such problem is to either precompute and store the answers 
to all queries in advance, which has optimal query time but prohibitive space ($s = |\cQ|$), or to 
store the raw database using optimal space ($s\sim n$) at the price of trivial query time ($t = n$). 
The obvious question is whether the underlying problem admits a better time-space trade-off.
Static data structure lower bounds aim to answer this question 
by proving unconditional lower bounds on this trade-off. %

The most compelling model for proving such lower bounds is the ``cell-probe"  
model \cite{Yao81}, in which a data structure is simply viewed as a table of $s$ memory cells ($w$-bit words)  
and query time is measured only by the \emph{number $t$ of memory accesses}  
(I/Os), whereas all computations on ``probed" memory cells are 
free of charge. This nonuniform\footnote{Indeed, a nonadaptive cell-probe data structure is 
essentially equivalent to an $m$-output 
depth-2 circuit with \emph{arbitrary gates}, ``width" $s$, and a \emph{bounded} top fan-in $t$, 
see \cite{JS11,BL15} and Section \ref{subsec_ckt_LB}.}
model of computation renders time-space trade-offs as purely 
information-theoretic question and thereby extremely powerful.
Unfortunately, the abstraction of this model also comes at a price: While a rather 
straight-forward counting argument \cite{Milt93} shows that \emph{most} static data structure problems  
with $m := |\cQ|$ queries indeed require either $t\geq n^{0.99}$ time or 
$s \geq m^{0.99}$ space (i.e., the na\"ive solutions are essentially optimal),  
the highest \emph{explicit} cell-probe lower bound known to date is  
\begin{align} \label{eq_cell_sampling_LB}
t \geq \Omega\left(\frac{\log (m/n)} {\log (s/n)}\right). 
\end{align}
In the interesting and realistic regime of polynomially many queries ($m = n^{O(1)}$), this 
yields a $t\gtrsim \log n$ lower bound on the query time of \emph{linear space} ($s=O(n)$) 
data structures for several natural problems, such as polynomial-evaluation, nearest-neighbor search 
and 2D range counting to mention a few \cite{Siegel04, patrascu08structures, PTW10, Lar12}. 
Proving an $\omega(\log n)$ cell-probe lower bound on \emph{any} explicit 
static problem in the linear space regime,  is a major open problem, and  
the trade-off in \eqref{eq_cell_sampling_LB} remains the highest static cell-probe lower bound known to date, 
even for \emph{nonadaptive} data structures  
(This is in sharp contrast to \emph{dynamic} data structures, where the restriction to nonadaptivity enables   
\emph{polynomial} cell-probe lower bounds \cite{BL15}). 

\smallskip 
In an effort to circumvent the difficulty of proving lower bounds in the cell-probe model, 
several restricted models of data structures have been studied over the years,   
e.g., the pointer-machine and word-RAMs~\cite{BoasPointerMachine90} as well as algebraic 
models, most notably, the \emph{group model} \cite{Fred81,Chazelle90, PatGroup07}. 
Since many important %
data structure problems involve \emph{linear queries} over the database 
(e.g., orthogonal range counting, partial sums, dictionaries,  
matrix-vector multiplication and polynomial evaluation 
to mention a few), it is natural to restrict the data structure to use only 
\emph{linear} operations as well. More formally, a static \emph{linear} data structure problem 
over a field $\F$ and input database $x \in \F^n$, is defined by an $m \times n$ matrix (i.e., a linear map) 
$M \in \F^{m\times n}$. The $m$ queries are the rows $M_i$ of $M$, and the answer to the $i$th query is 
$\langle M_i , x\rangle = (Mx)_i \in \F$. 
An $(s,t)$-\emph{linear data structure} for $M$  
is allowed to store $s$ arbitrary field elements in memory $P(x)\in \F^s$, and must compute 
each query $(Mx)_i$  as a \emph{$t$-sparse linear combination} of its memory state 
(we assume the word-size satisfies $w \geq \log |\F|$, see Section \ref{sec_lin_DS} for the complete details). 
This model is a special case of the \emph{static group model}, 
except here the group (field) is fixed in advance.\footnote{In the 
general (oblivious) group model, the input database consists of $n$ elements from a black-box (commutative) 
group, the data structure can only store and manipulate group elements through black-box group operations,   
and query-time ($t$) is measured by the number of algebraic operations (see e.g. \cite{Agarwal04, PatGroup07} for further details).} 

While the restriction to the group model has been fruitful for proving strong 
lower bounds on \emph{dynamic} data structures\footnote{In the dynamic setting, the data 
structure needs to maintain an \emph{online} sequence of operations while minimizing the 
number of memory accesses for update and query operations. In the group (linear) model, 
these constraints are essentially equivalent to a decomposition of a matrix $M = AB$ where \emph{both} 
$A$ and $B$ are sparse. In contrast, static lower bounds only require $A$ to be sparse, hence 
intuitively such decomposition is much harder to rule out.}  
(\cite{Agarwal04, PatGroup07, Lar14}), 
the static group model resisted this restriction as well, and \eqref{eq_cell_sampling_LB} 
remains the highest static lower 
bound even against nonadaptive linear data structures. 

\smallskip 
This paper shows that this barrier is no coincidence. We study linear data structures and show that 
proving super-logarithmic static lower bounds,  even against nonadaptive linear data structures with \emph{arbitrarily} small 
linear space $s = (1+\eps)n$, implies semi-explicit lower bounds on \emph{matrix rigidity}. Before stating our main results, we take a moment 
to introduce the notion of rigidity. %

\paragraph{Matrix rigidity. }
The notion of matrix rigidity was introduced by Valiant \cite{Valiant} as a possible approach for proving circuit lower bounds. We say that a matrix $A \in \F^{m\times n}$ is $(r,d)$-row rigid, if decreasing the rank of 
$A$ below $r$, requires modifying at least $d$ entries in \emph{some row} of $A$. In other words, for any $r$-dimensional subspace $U$ of $\F^n$, there exists a row in $A$ that is $d$-far (in Hamming distance) from $U$. We discuss a stronger notion of rigidity called \emph{'global rigidity'} later in the paper (requiring many rows of $A$ to be far from $U$) and prove a general reduction from one to the other (see Theorem~\ref{thm_ldc_worst_avg} below) which may be of independent interest. 

The best known bound on matrix rigidity of square matrices (for any rank parameter $r$) is $\Omega(\frac{n}{r}\log \frac{n}{r})$ \cite{friedman1993note,pudlak1994some,shokrollahi1997remark,Lokam09}\footnote{Goldreich and Tal~\cite{goldreich2016matrix} also give a ``semi-explicit" construction of rigid matrices 
which uses $O(n)$ bits of randomness. This construction has (global) rigidity $\Omega(\frac{n^2}{r^2\log{n}})$ for any $r\geq\sqrt{n}$, which improves on the classical bound for $r=o(\frac{n}{\log{n}\log\log{n}})$. In particular, since the required number of random bits is only linear, this construction is in $\mathbf{E^{NP}}$. 
Kumar and Volk~\cite{KV18} show how to construct an $(n^{0.5-\eps},\Omega(n^2))$-globally rigid matrix in subexponential time $2^{o(n)}$. }. Although matrix rigidity has attracted a lot of attention, this bound remains the best known for more than two decades. 

Matrix rigidity is also studied for rectangular matrices (introduced in \cite{APY09}), where we allow the number of rows $m$ to be larger than the number of columns. One can thus fix the parameters $r$ (typically $r \approx \eps n$) and $d$ (the sparsity) and try to minimize the number of rows in $A$. One can easily show that a random square matrix is highly rigid (say with $r$ and $d$ both  close to $n$), but the best explicit constructions of an $m \times n$ matrix which is $(\eps n, d$)-row-rigid requires 
$m = n\cdot 2^d$~\cite{APY09,SY11}.  This bound (and the related lower bound for square matrices which is even older) represent, to many experts  
in the field, a real barrier, and any improvement to it is likely to require substantially new ideas. Our results below show that this barrier can be surpassed 
if one can improve the currently known lower bounds on static data structures by a slightly super-logarithmic factor.

\subsection{Our results}
Our first main result is the following (see Theorem~\ref{thm:dsbound} for a formal statement):    

\begin{theorem}[Main Theorem, Informal] \label{thm_main_infromal}
A data structure lower bound of $t\geq \log^c n$ in the group (linear) model for computing 
a linear map $M \in \F^{m\times n}$, even against data structures with arbitrarily small 
linear space $s = (1+\eps)n$, yields an $(\eps n',d)$-row-rigid matrix $M' \in \F^{m\times n'}$ 
with $\eps n' \geq d\geq \Omega(\log^{c-1}{n})$.  
Moreover, if $M$ is explicit, then $M'\in\mathbf{P^{NP}}$.   
\end{theorem}

The premise of Theorem \ref{thm_main_infromal} would imply a (semi-explicit) construction of an 
$m \times n$ matrix which is $d \sim \log^{c-1}(m/n)$-rigid (i.e., requires modifying at least $d$ 
entries in some row to decrease the rank below, say, $n/4$).  
In comparison, the aforementioned best known explicit constructions only yield an $\Omega(\log(m/n))$-rigid matrix \cite{APY09,SY11}, 
which is only $\Omega(\log n)$ when $m=\poly(n)$. 
In particular, Theorem \ref{thm_main_infromal} asserts that proving a $t\geq \omega(\log^2 n)$ data structure 
lower bound against arbitrarily small linear space, would already yield an asymptotic improvement on (rectangular) 
rigid matrix construction.\footnote{Although here we state Theorem~\ref{thm_main_infromal} for the lowest probe complexity that is interesting, it actually gives a smooth trade-off: a lower bound on the linear data structure query time $t$ implies rigidity $\frac{t}{\log{n}}$.}

Theorem \ref{thm_main_infromal} indicates a ``threshold" in data structure 
lower bounds, since for \emph{succinct} data structures (which are constrained to use only $s = n + o(n)$ %
space), polynomial lower bounds ($t\geq n^{\eps}$) are known on the query time 
(e.g., \cite{gal:succinct,BL13}), even in the general \emph{cell-probe} model. 

Our second main result concerns implications of data structure lower bounds on \emph{square} matrix rigidity (see Theorem~\ref{thm:dsbound}, item $3$ for a formal statement):  

\begin{theorem}[Implications to Square Rigidity, Informal] \label{thm_square_infromal} 
For any $\delta>0$, a data structure lower bound of $t\geq \log^{3+\delta} n$ in the group (linear) model,  
for computing a linear map $M \in \F^{m\times n}$, 
even against arbitrarily small linear space $s = (1+\eps)n$, yields a \emph{square} 
matrix $M' \in \F^{n'\times n'}$ which is $\left(r , \omega(\frac{n'}{r}\log \frac{n'}{r})\right)$-rigid, 
for some $r=o(n)$. Moreover, if $M$ is explicit, then $M'\in\mathbf{P^{NP}}$.   
\end{theorem}

Since the highest rigidity bound known to date for square matrices (and any rank parameter $r$) 
is $\Omega(\frac{n'}{r}\log \frac{n'}{r})$ \cite{friedman1993note}, the premise of Theorem \ref{thm_square_infromal} would 
imply an asymptotic improvement over state-of-art lower bounds (the precise factor is given in the 
formal statement, see Theorem~\ref{thm:dsbound}). 

We note that the highest lower bounds known for both rigidity and linear data structures, are for \emph{error-correcting codes} 
such as the Vandermonde matrix (and furthermore, follow from essentially the same ``untouched-minor" technique). 
Our results demonstrate that this is not a coincidence. 

\smallskip

Our main result has further significant implications to other time-space regimes. 
In the succinct space regime, we show that any asymptotic improvement on the current best 
cell-probe lower bounds mentioned above, would yield improved rigidity bounds for near-square matrices, and vice versa. 
In particular, a corollary of this connection yields a logarithmic improvement on succinct lower bounds 
(in the group model): We exhibit an (explicit) data structure lower bound of $t\cdot r \geq \Omega(n\log(n/r))$ 
for linear data structures using space $s=n+r$, 
which is a logarithmic-factor improvement on the aforementioned bounds of \cite{gal:succinct,BL13} for a problem with linear number of queries $m=O(n)$.   

Finally, we show that `holy-grail' \emph{polynomial} ($t\geq n^{\delta}$) data structure lower bounds against near-trivial 
space $(s \sim m/\log\log m)$, would imply superlinear lower bounds for log-depth circuits (see Theorem~\ref{thm:clb}). 
We discuss all of these implications in Section \ref{sec_DS_rigidity}.

\paragraph{Related work.}
An independent concurrent work of Viola \cite{V18} shows a different connection between circuit lower bounds 
and cell-probe lower bounds, %
yet only in \emph{extremely high space} regimes $s \sim m$ (in fact, Theorem 4 in~\cite{V18}    
is very similar though formally incomparable to our Theorem~\ref{thm:clb}; Theorem 3 in~\cite{V18} is of independent interest).  
In contrast, our main results (Theorems~\ref{thm_main_infromal} and~\ref{thm_square_infromal}) focus on data structure lower 
bounds against arbitrarily small \emph{linear space} ($s= (1+\eps)n$, whereas $m=\poly(n)$), 
and prove a generic connection  between (linear) data structures and rigidity.  

Miltersen et al.~\cite[Theorem 4]{MNSW98} prove that any data structure problem with $n$ inputs and $m$ outputs over a field $\F$ of size $|\F|\geq\log{n}$ which can be solved by a read $O(1)$ branching program of polynomial size, can also be solved by a data structure with space $s=n^{O(1)}$ and query time $t=O(\frac{\log{m}}{\log{|\F|}-\log\log{n}})$.

\subsection{Technical Overview} %
It is not hard to see that a (nonadaptive) $(s,t)$-\emph{linear} data structure for a linear problem 
$M \in \F^{m\times n}$ is nothing but a factorization of $M$ as a product of two matrices $M=AB$, 
where $A$ is a $t$-sparse $(m \times s)$ matrix (with $\leq t$ non-zeros in each row), and $B$ is 
an arbitrary matrix with only $s$ rows (see Section \ref{sec_lin_DS}). 
As such, proving a lower bound on $(s,t)$ 
linear data structures is equivalent to finding an (explicit) matrix $M\in \F^{m\times n}$ which \emph{does not}  
admit such factorization, or equivalently, showing that $M$ is %
``\emph{sumset-evasive}", in 
the sense that the $m$ rows of $M$ (viewed as points $M_i \in \F^n$) are not contained %
in the $t$-span\footnote{I.e., the union of all $t$-dimensional subspaces 
generated by any \emph{fixed} set $S \subset \F^n$ of size $s$. We borrow the term ``sumset evasive" by analogy from 
additive combinatorics, but caution that this definition allows arbitrary \emph{linear combinations} and not just sums.}  
of \emph{any fixed} set of $s$ points in $\F^n$  (see Section \ref{sec_SE} below for the formal 
definition of $(s,t)$-sumset evasive sets).

In contrast, \emph{matrix rigidity} is the problem of finding an 
(explicit) matrix $M \in \F^{m\times n}$ which cannot be factorized as the \emph{sum} 
(rather than product) $M = A+B$ of a $t$-row-sparse matrix $A$ plus a \emph{low rank} 
matrix, say, $rk_\F(B)\leq r$. %
For brevity, unless otherwise stated, we say below that a matrix is $(r,d)$-rigid to mean 
that is is $d$-row-rigid, and that it is $t$-sparse to mean $t$-row-sparse. 

We establish  a new relationship between these two (seemingly disparate) factorizations. 
A key step in showing this relationship %
is to re-interpret the two factorization problems above as two (respective)
``geometric" measures on the \emph{column space} of $M$, i.e., viewing the matrix 
$M \in \F^{m\times n}$ as an $n$-dimensional subspace 
$V_M \subset \F^m$ spanned by its columns. 
Informally, the \emph{inner dimension} of $V_M$ %
is the maximal dimension $d_{M} \leq n$ 
of the \emph{intersection} of $V_M$ with any \emph{$t$-sparse subspace}\footnote{We say that a 
subspace $U \subseteq \F^m$ is 
\emph{$t$-sparse} if it is the column-space of a $t$-row-sparse matrix} $A$ of the 
same dimension $n$ (in other words, $V_M$ has small inner dimension $d_{M}(t)$ 
if it has low-dimensional intersection with any $n$-dimensional $t$-\emph{sparse} subspace, see 
Definition \ref{def:inner} below). The \emph{outer dimension} of $V_M$ %
is the minimal dimension $D_{M} \geq n$ of a $t$-sparse subspace $A$ that \emph{contains} $V_M$ 
(Definition \ref{def:outer}). We first prove the following characterization (Lemmas \ref{lem:inner_rigidity} and \ref{lem:eq}): 
\begin{itemize} 
\item $M$ is strongly\footnote{A matrix $M$ is strongly-rigid if it remains (row) rigid in \emph{any} 
basis of its column-space $V_M$, see Definition \ref{def:strong_rigidity}.} 
$(r,t)$-rigid if an only if $V_M$ has \emph{small inner dimension} ($d_{M}(t) < n-r$). 
\item $M$ is $(s,t)$ sumset-evasive if and only if $V_M$ has \emph{large outer dimension} ($D_{M}(t) > s$). 
\end{itemize}
(We note that the nontrivial direction of the first statement was already shown by \cite{PP06} for a subtly 
different notion of inner/outer dimensions, we provide an alternative proof in Lemma~\ref{lem:simplelemma} in Appendix~\ref{apx:a}).
In this terminology, proving that lower bounds on linear data structures  
imply lower bounds on (row) rigidity, is essentially equivalent to showing that \emph{large outer dimension implies small inner dimension} (perhaps of a related explicit matrix).
Indeed, our first main technical contribution is establishing the following relationship
between these two measures on \emph{submatrices} of $M$, 
which is the heart of the proof of Theorem \ref{thm_main_infromal}. 

\begin{lemma}[Large outer dimension implies small inner dimension, Theorem \ref{thm:main}, Informal] 
\label{lem_outer_inner_lemma_informal}
If $D_M(t) \geq (1+\eps)n$, there exists an $m\times n'$ submatrix 
$M' \subseteq M$ for which $d_{M'}(t/\log n) \leq (1-\eps)n'$.  
\end{lemma}
Indeed, by the characterization above, the last inequality implies that $M'$ is $(\eps n', t/\log n)$-rigid. 
The high level idea of the proof is a simple recursive procedure that, given a matrix $M$ with high 
outer dimension $D_M(t)$, `finds' a submatrix with low inner dimension. The heart of each iteration 
is as follows: If our current matrix (which is initially $M$ itself) is rigid (i.e., has low inner dimension, 
which can be checked with an \bf NP \rm oracle), then we are done. Otherwise, the \bf NP \rm oracle 
together with the characterization above, gives us a sparse subspace $V$ (of only $n$ dimensions) 
that has large intersection with the column space of $M$. After a change of basis (of the column space) 
we can essentially, partition the columns of $M$ into the part covered by $V$ and the remaining columns. 
We then apply the same argument on the remaining columns. At each iteration we `accumulate' an 
additional sparse $V$ (whose dimension is small -- merely the dimension of the residual space) and so, at the 
end, we must show that these can be pasted together to give a low-dimensional `cover' of the 
column-space $V$ of the original matrix $M$ (i.e., a small space data structure for $M$). 
Thus, the final sparsity grows by a factor proportional to the number of iterations, which is logarithmic.  
This implies that the process must end prematurely, resulting in a (semi-) explicit rigid submatrix.

\paragraph{Square Matrix Rigidity (Theorem \ref{thm_square_infromal}).} 
One caveat of Theorem \ref{thm_main_infromal} is that it may produce a highly skewed (rectangular) 
$m \times n'$ matrix, where $m \gg n'$ (this is because we only recurse on the column-space of $M$ but never 
on the row-space). While this already yields a significant improvement over current-best rectangular rigidity results 
(e.g., when $n'=\poly\log(n)$), this argument does not seem to imply anything about rigidity for \emph{square} matrices. 

A trivial  idea to turn the rectangular ($m \times n'$)-row-rigid matrix $M'$ produced by Lemma  
\ref{lem_outer_inner_lemma_informal} into a square matrix while preserving rigidity (at the expense of decreasing 
the relative rank parameter), is to ``stack" $m/n'$ identical copies of $M'$ side by side. 
Clearly, the rank of the resulting $m\times m$ matrix $M''$ remains unchanged (bounded by $n'$, 
so the relative rank parameter may decrease  significantly relative to $m$)
\footnote{The ratio between $n'$ and $m$ depends on the the postulated data structure lower 
bound on $M$, determining $n'$.}, 
but on the other hand, the hope is that $M''$ remains $\Omega(m)$-row-rigid. 
Indeed, Lemma \ref{lem_outer_inner_lemma_informal} guarantees that $M'$ is (say) $(n'/10,n'/2)$-row-rigid,  
and therefore in order to decrease the rank of $M''$ below $(n'/10)$, one would need to decrease the rank of 
\emph{each} of the $m/n'$ blocks below $n'/10$, which requires modifying $\sim (m/n') \cdot (n'/2) = m/2$ row entries in total.  The problem, of course, 
is that these rows may be \emph{different} in each block, which completely dooms the argument.  
Note that this is a direct consequence of working with \emph{row-rigidity}: 
If $M'$ were \emph{globally} rigid (i.e., at least $10\%$ of its rows need to be modified in at least 
$\sim t$ entries in order to decrease the rank below $n'/10$), this simple trick would have gone through 
flawlessly. 
In order to bypass this obstacle, we prove the following \emph{black-box} reduction from row-rigidity to global-rigidity. Our reduction uses Locally Decodable Codes (LDCs) in a way reminiscent of the way LDCs are used in worst-case to average-case hardness amplification  results \cite{IW01}.

\begin{theorem}[Row to Global Rigidity] \label{thm_ldc_worst_avg}
Let $E : \F^m \mapsto \F^{m'}$ be a \emph{linear} $q$-query locally decodable code (LDC) against 
constant noise $\delta$,  
and let $E(A) \in \F^{m' \times n}$ be the application of $E$ to each column of $A\in \F^{m\times n}$. 
Then if $A$ is $(r,t)$-row-rigid, then $E(A)$ is $(r , \delta t m'/q)$-globally rigid. 
\end{theorem}

To prove the theorem, we use the following 
crucial property of \emph{linear} LDCs, originally observed by Goldreich et. al \cite{GKST02}: 
If $E : \F^m \mapsto \F^{m'}$ is a $q$-query linear LDC (applied on the columns of $A\in \F^{m\times n}$), 
then for \emph{any} subset $S$ of at most $\delta m'$ rows 
of $E(A)$, and any row $A_j$ of $A$, there exist $q$ rows of $E(A)$ that lie \emph{outside} $S$ and span $A_j$. (See Section \ref{sec:global} for a formal argument). 
Now, suppose towards contradiction, that $E(A)$ is not $(r , d)$-globally rigid, for $d := t\delta m'/2q$. 
This means that there is some $r$-dimensional subspace $L \subset \F^n$ which is 
at most $(t\delta/2q)$-far (in Hamming distance) from an average row of $E(A)$, hence by a Markov argument, 
at most $\delta m'$ of the rows of $E(A)$   are $> (t/2q)$-far 
from $L$. Let $\cB$ denote this set of rows. Since $|\cB|\leq \delta m'$, 
the LDC property above asserts that \emph{every row of $A$} is a linear combination of at most $q$ rows in 
$E(A)\setminus \cB$, each of which is $(t/2q)$-close to $L$ by definition of $\cB$. 
But this means that \emph{every} row of $A$ 
is at most $(t/2q) \cdot q = t/2$ far from $L$, which is a contradiction since $A$ was assumed to be $t$-row-rigid 
(hence there must be at least \emph{one} row that is $t$-far from $L$). The complete proof can be found in Theorem~\ref{thm:global}. 

Since there are explicit linear $q=\log^{1+\eps}(n)$-query LDCs with polynomial rate ($m' \approx m^{1/\eps}$), 
Theorem \ref{thm_ldc_worst_avg} now completes the proof of Theorem \ref{thm_square_infromal} using the 
aforementioned program (stacking copies of $M'$ next to each other), at the price of an extra logarithmic loss in sparsity. 
To the best of our knowledge, Theorem \ref{thm_ldc_worst_avg} establishes the first nontrivial relationship 
between rectangular and square (row) rigidity, hence it may be of independent and general interest to 
rigidity theory (as the notion of row-rigidity is potentially much weaker than global rigidity). 
We also remark that Theorem \ref{thm_ldc_worst_avg} applies in the same way to reduce 
worst-case to average-case \emph{data-structure} lower bounds for linear problems.

\section{Discussion: Group Model vs. Cell-Probe Model}

Our results (and in particular Theorems \ref{thm_main_infromal} and \ref{thm_square_infromal}) do not  
have formal implications on the hardness of proving cell-probe lower bounds on \emph{nonlinear} data structure problems, 
as our arguments (and the notion of rigidity itself) inherently rely on the linear structure of the underlying problem (i.e., a matrix).
On the other hand, many linear data structure problems are believed to require very high (polynomial)  
query time with near-linear space, even in the general cell-probe model (some examples are simplex or 
ball range counting problems in $d=O(1)$ dimensions \cite{AE99}, or multivariate polynomial evaluation). 
Since cell-probe lower bounds clearly imply group-model lower bounds, our 
rigidity implications apply for cell-probe model on such problems.

On a related note, it is tempting to conjecture that the optimal data structures for linear problems 
can be achieved using only linear operations (without a significant blowup in space or time). 
This was conjectured by Jukna and Schnitger (see Conjecture 1 in \cite{JS11}), in circuit-complexity terminology. 
Interestingly, there are surprising examples of very efficient \emph{non-linear} data structures 
for linear problems, which have no known linear counterparts (\cite{KU08}). Alas, our result 
implies that in certain parameter regimes, refuting this conjecture implies (semi-explicit) rigidity 
lower bounds. The conjecture of \cite{JS11} remains a fascinating open problem.

\section{Organization} 
We begin with some formal definitions and the necessary background in Section \ref{sec:prelim}, 
which set up the stage for our  results.  
In Section \ref{sec:building_blocks} we prove the main technical building block which relates the inner and outer 
dimensions of a matrix (Theorem \ref{thm:main}). Section \ref{sec:global} contains the reduction from row rigidity 
to global rigidity. In Section \ref{sec_DS_rigidity} we use these two technical building blocks to derive Theorems 
\ref{thm_main_infromal} and \ref{thm_square_infromal}, as well as further connections between linear data 
structure lower bounds in various time-space regimes, and matrix rigidity. Implications to arithmetic log-depth 
circuits are also discussed.  

\section{Setup and Preliminaries} 
\label{sec:prelim}

\subsection{Linear Data Structures (Static Group Model)} \label{sec_lin_DS}

A \emph{linear} data structure problem with $|\cQ| = m$ queries over a field $\F$ and an input database of $n$ 
elements is defined by an $m \times n$ matrix (i.e., a linear map) $V \in \F^{m\times n}$. %
The queries are the rows 
$V_i$ of $V$, and for any input database $x \in \F^n$, the answer to the $i$th query is given by $ \langle V_i , x\rangle = (Vx)_i \in \F$. 

An $(s,t)$ nonadaptive data structure $\cD$ for %
the problem $V$ in the cell-probe model is a pair $\cD = (P,Q)$, where $P$ is a \emph{preprocessing} 
function $P : \F^n \mapsto\F^s$ %
that encodes the database $x\in \F^n$ into $s$ memory \emph{cells},
and a query algorithm $Q : \F^s \mapsto \F^m$ that correctly answers 
every query of $V$ by \emph{probing} at most $t$ memory cells\footnote{the indices of memory cells are only function of the query index}, i.e.,   
such that $Q(P(x)) = (Vx)_i$, for every $x \in \F^n$  and every query $i\in [m]$.

$\cD$ is a \emph{linear data structure} for the (linear) problem $V$ if both $P$ and $Q$ only compute linear functions 
over $\F$ .%
We observe that it suffices to require $Q$ to be linear: if a linear problem $V$ is solved by a data structure $\cD$ with 
a linear \emph{query} function $Q$, then $\cD$ can be transformed into an equivalent data structure with the same parameters 
$s$ and $t$ where \emph{both} $P$ and $Q$ are linear.

\begin{restatable}[Lemma 2.5~\cite{JS11}, Ex. 13.7~\cite{J12}]{proposition}{linearization}
\label{prop:linearization}
Given an $(s,t)$-data structure $\cD$ computing a linear transformation $V\bm{x}$ for $V\in \F^{m\times n}, x\in\F^n$ with linear query function $Q$, one can efficiently construct an equivalent $(s,t)$-data structure where \emph{both} the query function $Q$ and the preprocessing function $P$ are linear.
\end{restatable}

\subsection{Inner and Outer Dimensions}
We state our main technical results in terms of Paturi-Pudl{\'a}k dimensions~\cite{PP06,Lokam09}, and then show that they imply new connections between data structure lower bounds and matrix rigidity. While Paturi-Pudl{\'a}k dimensions are defined w.r.t. column sparsity, for our applications we need to consider %
an analogous definition w.r.t. row sparsity (this difference is important in this context).

\begin{definition}[Sparse subspaces]
A matrix $M\in\F^{m\times n}$ is \emph{$t$-globally sparse} if it has $t$ non-zero elements, and $M$ is \emph{$t$-row sparse} if each of its \emph{rows} has at most $t$ non-zero entries. A subspace $V\subseteq\F^m$ is \emph{$t$-sparse} if it is the \emph{column space} of a $t$-row sparse matrix.
\end{definition}

\begin{definition}[Inner dimension~\cite{PP06}]\label{def:inner}
Let $V\subseteq\F^m$ be a subspace, and $t$ be a sparsity parameter. Then the \emph{inner dimension} $d_V(t)$ of $V$ is
\begin{align*}
d_V(t) = \max_U{\left\{\dim(V \cap U)\colon \dim(U)\leq\dim(V),\; U \text{ is } t\text{-sparse}\right\} } \, .
\end{align*}
\end{definition}
\begin{definition}[Outer dimension~\cite{PP06}]\label{def:outer}
Let $V\subseteq\F^m$ be a subspace, and $t$ be a sparsity parameter. Then the \emph{outer dimension} $D_V(t)$ of $V$ is
\begin{align*}
D_V(t) = \min_U{\left\{\dim(U)\colon V\subseteq U,\; U \text{ is } t\text{-sparse}\right\} } \, .
\end{align*}
\end{definition}
By abuse of notation, for a matrix $M\in\F^{m\times n}$ we denote by $d_M(s)$ and $D_M(s)$ the inner and outer dimensions of the column space of $M$.

\subsection{Sumset Evasive Sets} \label{sec_SE}
For an integer $t$ and a set of points $S\subseteq\F^n$, 
$tS$ denotes the $t$-span of $S$, i.e., the union of all $t$-sparse linear combinations %
of $S$ : 
\begin{align*}
tS := \{w_1\cdot s_1 + \ldots + w_t\cdot s_t \; \colon \forall i, w_i\in\F, s_i \in S \} \, .\text{\footnotemark}
\end{align*}
\footnotetext{We note that the $t$-sum of $S$ is often defined as the set $\{s_1 + \ldots + s_t \colon \forall i, s_i \in S \}$. 
We abuse the notation by using the term $t$-sum for all linear combinations of length $t$ of the vectors from $S$.}

\begin{definition}[Sumset evasive sets]
For integers $s$ and $t$ we say that a set $M\subseteq\F^n$ of size $|M|=m$ is $(s,t)$-sumset evasive if for any set $S \subseteq \F^n$ 
of size $|S|=s$, it holds that\text{\footnotemark} 
\begin{align*}
|t S \cap M| < m \, . 
\end{align*}
\end{definition}
\footnotetext{We shall see that the definition of sumset evasive sets exactly captures the hardness of linear data structure problems. One can extend the definition of sumset evasive sets to capture the hardness of approximating a linear problem by a linear data structure. Since the main focus of this work is exact data structures, we omit this extended definition.}

The next lemma asserts that linear data structure lower bounds, sumset evasive sets and subspaces of high outer dimension are all equivalent.

\newcommand{\Implies}[2]{$\text{\ref{#1}}\implies\text{\ref{#2}}$}

\begin{lemma}\label{lem:eq}
Let $M\subseteq\F^n$ be a set of size $|M|=m$,  let $A\in\F^{m\times n}$ be a matrix composed of the vectors of $M$, and let $V\subseteq\F^m$ be the column space of $A$. The following are equivalent:
\begin{enumerate}[label=(\arabic*),ref=(\arabic*)]
   \item There is an $(s,t)$ linear data structure computing $A$. \label{statement1}
   \item $D_V(t)\leq s$. \label{statement2}
   \item $M$ is \emph{not} $(s,t)$-sumset evasive. \label{statement3}
\end{enumerate}
\end{lemma}
\begin{proof} 
\Implies{statement2}{statement1}: Since $D_V(t)\leq s$, there exists a $t$-sparse subspace $U\subseteq\F^m$ of $\dim(U)\leq s$ such that $V\subseteq U$. Then let $Q\in\F^{m\times s}$ be a $t$-row sparse matrix whose columns generate $U$. Since $V\subseteq U$, each column of $A$ is a linear combination of columns from $Q$. Therefore, there exists a matrix $P\in\F^{s\times n}$ such that $A=Q \cdot P$. We show that there exists an $(s,t)$ linear data structure $\cD$ which computes $A$. Indeed, let the preprocessing function of $\cD$ be the linear transformation defined by $P$, and let the query algorithm be the linear function defined by $Q$. Since $Q$ is $t$-row sparse, and $P\in\F^{s\times n}$, $\cD$ is an $(s,t)$ linear data structure.

\Implies{statement3}{statement2}: Since $M$ is not $(s,t)$-sumset evasive, there exists a set $S\subseteq\F^n$ of size $|S|=s$ such that $M \subseteq tS$. Let $P\in\F^{s\times n}$ be a matrix composed of the vectors of $S$. Since $M \subseteq tS$, there exists a $t$-row sparse matrix $Q\in\F^{m\times s}$ such that $A=Q\cdot P$. Let $U\subseteq \F^m$ be the column space of $Q$. We have that $V \subseteq U$ and $\dim(U)\leq s$.

\Implies{statement1}{statement3}: Let $\cD$ be an $(s, t)$ linear data structure which computes $A$. Let $P\in\F^{s\times n}$ be the linear transformation computed by its preprocessing function, and $Q\in\F^{m\times s}$ be the linear transformation computed by its query function. Let $S\subseteq\F^n$ be the set of $s$ rows of $P$. Since $Q$ is $t$-row sparse, the set $tS$ contains the set $M$ which contradicts sumset evasiveness of $M$.

\end{proof}

\subsection{Rigidity} \label{sec_prelim_rigidity}
\begin{definition}[Rigidity]\label{def:rigidity}
A matrix $M\in\F^{m\times n}$ is \emph{$(m, n, r, t)$-row rigid} if any matrix which differs from 
$M$ in at most $t$ elements in each row, has rank at least $r$. $M$ is \emph{$(m, n, r, t)$-globally rigid} if any matrix which differs from $M$ in at most $t$ elements has rank at least $r$.
\end{definition}

In other words, $M$ is rigid if it \emph{cannot} be written as a sum $M = A + B$ of 
a sparse matrix $A$ and a low rank matrix $B$. %

Now we define a stronger notion of rigidity which is invariant under basis changes.
\begin{definition}[Strong rigidity]\label{def:strong_rigidity}
A matrix $M\in\F^{m\times n}$ is \emph{$(m, n, r, t)$-strongly row rigid} if $MT$ is $(m, n, r, t)$-row rigid for any invertible matrix $T\in\F^{n\times n}$. Similarly, $M\in\F^{m\times n}$ is \emph{$(m, n, r, t)$-strongly globally rigid} if $MT$ is $(m, n, r, t)$-globally rigid for any invertible matrix $T\in\F^{n\times n}$.\footnote{We remark that while strong rigidity is interesting for rectangular matrices, and many of the known constructions of rigid matrices are actually strongly rigid (see, e.g., Proposition~\ref{prop:strong} in Appendix~\ref{apx:a}), this definition is meaningless for square matrices. Indeed, any full-rank matrix $M\in\F^{n\times n}$ becomes non-rigid after multiplication by $M^{-1}$.} 
\end{definition}

Now we give two equivalent definitions of strong rigidity. Friedman~\cite{friedman1993note} defines strong rigidity in the same way as inner dimension. We show that this definition is equivalent to the definition above. The equivalence between strong rigidity and inner dimension (which is a modified version of Proposition~3 from~\cite{PP06}) will play a key role in our proof of Theorem~\ref{thm_main_infromal}. It asserts 
that if a matrix $M$ is non-rigid, then there must be some sparse subspace with a significant intersection with the column space of $M$.

\begin{lemma}[Inner dimension is equivalent to strong rigidity]\label{lem:inner_rigidity}
Let $M\in\F^{m\times n}$ be a matrix of rank $\rk(M)=n$, and $V\subseteq\F^m$ be its column subspace of $\dim(V)=n$. Then the following are equivalent:
\begin{enumerate}
\item $M$ is $(m, n, r, t)$-strongly row rigid. \label{def1}
\item $d_V(t)\leq \rk(M)-r$. \label{def2}
\item $V$ is \emph{not} contained in a subspace of the form $A+B$ such that $A\subseteq\F^m$ is a $t$-sparse subspace of dimension $\dim(A)\leq n$, and $B\subseteq\F^{m}$ is a subspace of dimension $\dim(B)<r$. \label{def3}
\end{enumerate}
\end{lemma}

\begin{proof}
\Implies{def1}{def2}: 
Assume that $d_V(t)> \rk(M)-r$. Then, by Definition~\ref{def:inner}, there exists a $t$-sparse subspace $U\subseteq\F^m$ of $\dim(U)\leq \rk(M)$, such that $\dim(U\cap V)> \rk(M)-r$. Thus, there exists a subspace $W\subseteq\F^m$ of $\dim(W)< r$ such that $V=U+W$. Let $A$ be a $t$-row sparse basis matrix of the subspace $U$. Then there is an invertible matrix $T\in\F^{n\times n}$ such that $M=AT+B$, where $\rk(B)<r$. Therefore, $MT^{-1}=A+BT^{-1}$ is not an $(m, n, r, t)$-row rigid, and $M$ is not strongly row rigid.

\Implies{def2}{def3}:
Assume that there exist subspaces $A, B \subseteq \F^{m}$, where $A$ is $t$-sparse, $\dim(A)\leq n$ and $\dim(B)< r$, such that $V\subseteq A+B$. From $V\subseteq A+B$ we have that $A+V\subseteq A+B$ which gives us that
\begin{align*}
\dim(A)+\dim(V)-\dim(A\cap V) = \dim(A+V)\leq\dim(A+B)\leq \dim(A)+\dim(B) \,
\end{align*}
and 
\begin{align*}
\dim(A\cap V) \geq \dim(V)-\dim(B)> \rk(M)-r \, ,
\end{align*}
which implies that $d_V(t)\geq \dim(A\cap V) > \rk(M)-r$.

\Implies{def3}{def1}: 
Assume that $M$ is not $(m, n, r, t)$-strongly row rigid. Then there exist $T\in\F^{n\times n}, A,B\in\F^{m\times n}$ such that $MT=A+B$, $A$ is $t$-sparse, and $\rk(B)<r$. Then $V$ is contained in the span of columns of $A$ and $B$. Since $A$ is $t$-sparse, and $\rk(B)<r$, this contradicts the assumption that $V$ is not a subspace of such a sum.
\end{proof}

\subsection{Circuit Lower Bounds} \label{subsec_ckt_LB}
A long-standing open problem in circuit complexity is to prove a super-linear lower bound on the size of circuits of depth $O(\log{n})$ computing an explicit function \cite[Frontier~3]{Valiant, AB2009}. The same question remains open for linear circuits (i.e., circuits where each gate computes a linear combination of two of its inputs) computing an explicit linear map $f\colon\{0,1\}^n\to\{0,1\}^n$. Using a classical graph-theoretic result~\cite{erdos1975sparse}, Valiant~\cite{Valiant} reduced this problem to a problem about depth-$2$ circuits of a special kind: there are only $O(n/\log\log{n})$ gates in the middle layer which depend on the $n$ inputs, and each output gate depends on $n^{\eps}$ input and middle layer gates (for an arbitrary constant $\eps$). Note that a static data structure can be thought of as a depth-$2$ circuit with $n$ inputs, $m$ outputs, $s$ middle layer gates which depend on inputs, where each output depends on $t$ gates in the middle layer. Figures~\ref{fig:ds} (a) and~(b) illustrate the depth-$2$ circuits corresponding to static data structures and Valiant's reduction.

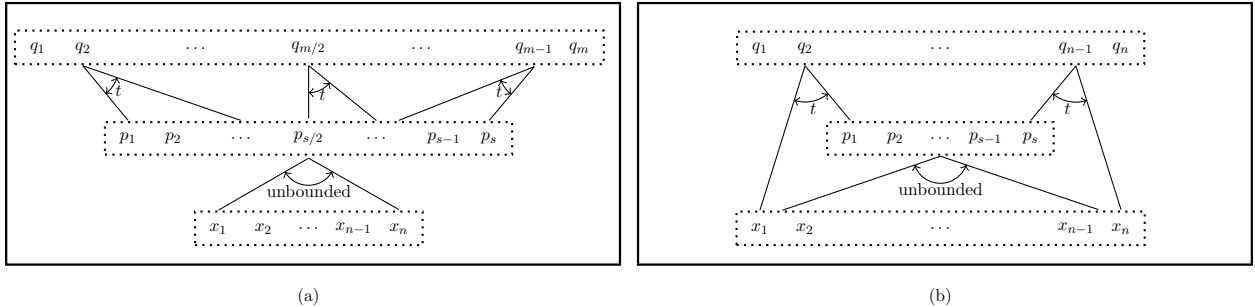
\begin{figure}
\begin{tikzpicture}
\begin{scope}%
[every node/.style={draw=black,circle,minimum size=2mm,black,inner sep=.5mm},
line width=.3mm,
scale=0.6,
transform shape]
\draw[draw] (0.3,0.2) rectangle (13.9,6);
\draw[draw=none, use as bounding box] (0.3,-0.7) rectangle (13.9,6);
  \foreach \pos/\name/\text in {{(5,1)/x1/x_1}, {(6,1)/x2/x_2}, {(7,1)/xdots/\ldots}, {(8,1)/xn1/x_{n-1}}, {(9,1)/xn/x_n}}
    \node[draw=none] (\name) at \pos {$\text$};
  \foreach \pos/\name/\text in {{(3,3)/p1/p_1}, {(4,3)/p2/p_2},  {(5.5,3)/pdots1/\ldots}, {(7,3)/ps2/p_{s/2}}, {(8.5,3)/pdots2/\ldots}, {(10,3)/ps1/p_{s-1}}, {(11,3)/ps/p_s}}
    \node[draw=none] (\name) at \pos {$\text$};
  \foreach \pos/\name/\text in {{(1,5)/q1/q_1}, {(2,5)/q2/q_2}, {(4.5,5)/qdots1/\ldots}, {(7,5)/qm2/q_{m/2}}, {(9.5,5)/qdots2/\ldots}, {(12,5)/qm1/q_{m-1}}, {(13,5)/qm/q_m}}
    \node[draw=none] (\name) at \pos {$\text$};

 \draw[thick,dotted]     ($(x1.south west)+(-0.3,-0.15)$) rectangle ($(xn.north east)+(0.3,0.15)$);
 \draw[thick,dotted]     ($(p1.south west)+(-0.3,-0.15)$) rectangle ($(ps.north east)+(0.3,0.15)$);
 \draw[thick,dotted]     ($(q1.south west)+(-0.3,-0.15)$) rectangle ($(qm.north east)+(0.3,0.15)$);
 
\coordinate (x1p) at ($(x1.north)+(0,0.1)$);
\coordinate (xnp) at ($(xn.north)+(0,0.1)$);
\coordinate (ps2p) at ($(ps2.south)+(0,-0.0)$);

\coordinate (p1p) at ($(p1.north)+(0,0.1)$);
\coordinate (pdots1p) at ($(pdots1.north)+(0,0.1)$);
\coordinate (q2p) at ($(q2.south)+(0,-0.1)$);

\coordinate (psp) at ($(ps.north)+(0,0.1)$);
\coordinate (pdots2p) at ($(pdots2.north)+(0,0.1)$);
\coordinate (qm1p) at ($(qm1.south)+(0,0.1)$);

\coordinate (ps2p2) at ($(ps2.north)+(0,-0.0)$);
\coordinate (qm2p) at ($(qm2.south)+(0,0.1)$);
\coordinate (pdots2p2) at ($(pdots2.north)+(0.5,0.1)$);

\pic [draw, <->,
      angle radius=6mm, thin, angle eccentricity=1.2,
      "unbounded" {draw=none, thin}] {angle = x1p--ps2p--xnp};

\draw[thin,-] (x1p) -- (ps2p);
\draw[thin,-] (ps2p) -- (xnp);

\pic [draw, <->,
      angle radius=8mm, thin, angle eccentricity=1.2,
      "$t$" {draw=none, thin}] {angle = p1p--q2p--pdots1p};

\draw[thin,-] (p1p) -- (q2p);
\draw[thin,-] (q2p) -- (pdots1p);

\pic [draw, <->,
      angle radius=8mm, thin, angle eccentricity=1.2,
      "$t$" {draw=none, thin}] {angle = pdots2p2--qm1p--psp};

\draw[thin,-] (psp) -- (qm1p);
\draw[thin,-] (qm1p) -- (pdots2p2);

\pic [draw, <->,
      angle radius=6mm, thin, angle eccentricity=1.2,
      "$t$" {draw=none, thin}] {angle = ps2p2--qm2p--pdots2p};

\draw[thin,-] (ps2p2) -- (qm2p);
\draw[thin,-] (pdots2p) -- (qm2p);

\node[draw=none] at (7,-0.5) {(a)};
\end{scope}

\begin{scope}
[every node/.style={draw=black,circle,minimum size=2mm,black,inner sep=.5mm},
line width=.3mm,
xshift=84mm,
scale=0.6,
transform shape]
\draw[draw] (0.3,0.2) rectangle (13.9,6);
\draw[draw=none, use as bounding box] (0.3,-0.7) rectangle (13.9,6);
  \foreach \pos/\name/\text in {{(3,1)/x1/x_1}, {(4,1)/x2/x_2}, {(7,1)/xdots/\ldots}, {(10,1)/xn1/x_{n-1}}, {(11,1)/xn/x_n}}
    \node[draw=none] (\name) at \pos {$\text$};
  \foreach \pos/\name/\text in {{(5,3)/p1/p_1}, {(6,3)/p2/p_2},  {(7,3)/pdots/\ldots}, {(8,3)/ps1/p_{s-1}}, {(9,3)/ps/p_{s}}}
    \node[draw=none] (\name) at \pos {$\text$};
  \foreach \pos/\name/\text in {{(3,5)/q1/q_1}, {(4,5)/q2/q_2}, {(7,5)/qdots/\ldots}, {(10,5)/qn1/q_{n-1}}, {(11,5)/qn/q_n}}
    \node[draw=none] (\name) at \pos {$\text$};

 \draw[thick,dotted]     ($(x1.south west)+(-0.3,-0.15)$) rectangle ($(xn.north east)+(0.3,0.15)$);
 \draw[thick,dotted]     ($(p1.south west)+(-0.3,-0.15)$) rectangle ($(ps.north east)+(0.3,0.15)$);
 \draw[thick,dotted]     ($(q1.south west)+(-0.3,-0.15)$) rectangle ($(qn.north east)+(0.3,0.15)$);

\coordinate (x1p2) at ($(x1.north)+(0.5,0.1)$);
\coordinate (xnp2) at ($(xn.north)+(-0.5,0.1)$);
\coordinate (pdotsp) at ($(pdots.south)+(0,-0.1)$);

\coordinate (x1p) at ($(x1.north)+(0,0.1)$);
\coordinate (q2p) at ($(q2.south)+(0,-0.1)$);
\coordinate (p1p) at ($(p1.north)+(0,0.1)$);

\coordinate (xnp) at ($(xn.north)+(0,0.1)$);
\coordinate (qn1p) at ($(qn1.south)+(0,0.08)$);
\coordinate (psp) at ($(ps.north)+(0,0.1)$);

\pic [draw, <->,
      angle radius=6mm, thin, angle eccentricity=1.2,
      "unbounded" {draw=none, thin}] {angle = x1p2--pdotsp--xnp2};

\draw[thin,-] (x1p2) -- (pdotsp);
\draw[thin,-] (pdotsp) -- (xnp2);

\pic [draw, <->,
      angle radius=8mm, thin, angle eccentricity=1.2,
      "$t$" {draw=none, thin}] {angle = x1p--q2p--p1p};

\draw[thin,-] (x1p) -- (q2p);
\draw[thin,-] (q2p) -- (p1p);

\pic [draw, <->,
      angle radius=8mm, thin, angle eccentricity=1.2,
      "$t$" {draw=none, thin}] {angle = psp--qn1p--xnp};

\draw[thin,-] (xnp) -- (qn1p);
\draw[thin,-] (qn1p) -- (psp);

\node[draw=none] at (7,-0.5) {(b)};
\end{scope}
\end{tikzpicture}
\caption{(a)~A (nonadaptive) static data structure as a depth-$2$ circuit. The $n$ input nodes feed $s\geq n$ memory cells, and we do not pose any restrictions on linear functions computed in memory cells. Each query (or output gate) depends only on $t$ memory cells. In a typical scenario, $t$ is as low as $t=(\log{n})^{O(1)}$ or $t=n^{\eps}$.
(b)~Depth-$2$ circuit resulting from Valiant's reduction. The $n$ inputs feed only $s=O(n/\log\log{n})$ middle layer gates. Again, we do not pose any restrictions on the linear functions computed in the middle layer gates. Each of the $n$ output gates depends only on $t<n^{\eps}$ inputs and middle layer gates.}
\label{fig:ds}
\end{figure}

From Valiant's reduction (see Figure~\ref{fig:ds}~(b)) one can conclude that if a linear-size log-depth linear circuit computes a linear map $M\in\F^{n\times n}$, then $M$ can be written as $M=A+C\cdot D$, where $A, C$, and $D$ encode the dependence of outputs on inputs, the dependence of outputs on the middle layer, and the dependence of the middle layer on inputs, respectfully.
Note that since every output has fan-in $t$, we can conclude that the matrices $A$ and $C$ are $t$-sparse. Formally, Valiant gave the following decomposition:
\begin{theorem}[\cite{Valiant}]\label{thm:valiant}
Let $m\geq n$. For every $c,\eps>0$, there exists $\delta>0$ such that any linear map $M\in\F^{m\times n}$ computable by a circuit of size $cm$ and depth $c\log{m}$, can be decomposed as
\begin{align*}
M=A+C\cdot D \, ,
\end{align*}
where $A\in\F^{m\times n}, C\in\F^{m\times s}, D\in\F^{s\times n}$, $A$ and $C$ are $t$-sparse. There are two decompositions:
\begin{itemize}
\item  $s=\eps m, t=2^{(\log m)^{1-\delta}}$;
\item $s=\frac{\delta m}{\log\log{m}}, t=m^{\eps}$.
 \end{itemize}
\end{theorem}

In particular, from the dimensions of $C$ and $D$, $C\cdot D$ has rank at most $s$. Thus, $M=A+B$ for a $t$-sparse $A$ and $\rk(B)\leq s$. 
\begin{corollary}
An $(n,n,\eps n, n^{\delta})$-row rigid matrix (for arbitrary constants $\eps,\delta$) does not have linear-size log-depth circuits.
\end{corollary}

The best known (row) rigidity lower bound for the regime of $r=\eps n$ is only $t\geq\Omega(1)$. If we relax the requirement of matrices to be (almost) square, then for $m=\poly(n)$ we know examples of $m\times n$ matrices with $t\geq\Omega(\log{n})$~\cite{APY09,SY11}.

The problem of finding non-trivial circuit lower bounds and rigid matrices is open not only in $\mathbf P$, but also in larger uniform classes like $\mathbf{P^{NP}}$ or even $\mathbf{E^{NP}}$.

\section{Main Building Blocks} \label{sec:building_blocks}
This section contains our two main tools for converting data structure lower bounds into rigidity lower bounds. 
In Section~\ref{sec:outer_inner}, we show that a rectangular matrix $M\in\F^{m\times n}$ which is hard for linear data structures contains a rectangular submatrix of high row-rigidity. In Section~\ref{sec:global}, we show that a rectangular matrix of high row-rigidity can be transformed in a rigid \emph{square} matrix (with some 
loss in the relative rank parameter but almost no loss in the relative sparsity parameter).

\subsection{Connection Between Outer and Inner Dimensions}\label{sec:outer_inner}
In this section we shall prove that every matrix either has small outer dimension or contains a matrix of small inner dimension. From Lemmas~\ref{lem:eq} and~\ref{lem:inner_rigidity}, this implies that every matrix which cannot be computed by efficient data structures (has large outer dimension) contains a rigid submatrix (submatrix of low inner dimension).
We start with the following auxiliary lemma.
\begin{lemma}\label{lem:aux}
Let $m,n,k$ be positive integers.
If $M\in\F^{m\times n}$ has $d_M(t)\geq \rk(M)-k$, then $M$ can be decomposed as
\begin{align*}
M=A\cdot B+M'\cdot C \, ,
\end{align*}
where $M'\in\F^{m\times k}$ is a submatrix of $M$,
$A\in\F^{m\times n}$ is $t$-row sparse,
$B\in\F^{n\times n}$, $C\in\F^{k \times n}$.

Moreover, if $\F$ is a finite field of size $2^{m^{O(1)}}$, such a decomposition can be found in time $\poly(n,m)$ with an $\mathbf{NP}$ oracle.
\end{lemma}

\begin{proof}
Let $V$ be the column space of $M$. By Definition~\ref{def:inner}, there exists a $t$-sparse subspace $U\subseteq\F^m, \dim(U)\leq \dim(V)\leq n$, such that $\dim(V\cap U) \geq \rk(M)-k $. Let $A\in\F^{m\times n}$ be a \emph{$t$-row sparse} matrix generating $U$. 

Let us now extend $A$ with at most $k$ column vectors from $M$ to generate the column space of $M$, and let $M'\in\F^{m \times k}$ be a matrix formed by these columns. 
Since the columns of $A\in\F^{m\times n}$ and $M'\in\F^{m \times k}$ together generate the column space of $M$, there exist matrices $B \in \F^{n\times n}, C\in \F^{k\times n}$ such that 
\begin{align*}
M = A\cdot B + M'\cdot C \, .
\end{align*}

Let $\operatorname{INNER-DIM}$ be the language of triples $(M, t, d)$ such that the matrix $M$ has inner dimension (with the sparsity parameter $t$) $d_M(t)\geq d$. Since there is a polynomial-size witness (a $t$-row sparse matrix whose column space intersects with $M$ in at least $d$ dimensions) which can be verified in polynomial time, $\operatorname{INNER-DIM}\in\mathbf{NP}$. Now we apply the standard search-to-decision reduction. Namely, we define $\mathbf{NP}$ languages so that we could use binary search to find each coordinate of a matrix $A$ witnessing high inner dimension of $M$. For a field of size $2^{m^{O(1)}}$, this can be done with $\poly(m,n)$ queries to the $\mathbf{NP}$ oracle. Now, we can just use Gaussian elimination (running in time $\poly(n,m$)) to find a matrix $M'$, and then the matrices $B$ and $C$.
\end{proof}

We are now ready to present the main result of this section.
\begin{algorithm}
\caption{Find a Submatrix with Low Inner Dimension in a Matrix with High Outer Dimension}
\label{alg:main}
\hspace*{\algorithmicindent} \textbf{Input:} Parameters $\eps, k, t$, and a matrix $M\in\F^{m\times n}$ with $D_M(tk +n\eps^k)\geq\frac{n}{1-\eps}$.\\
\hspace*{\algorithmicindent} \textbf{Output:} A submatrix $M'\in\F^{m\times n'}$ of $M$ with $d_{M'}(t)<\rk(M')- \eps n'$ and $n'\geq n\eps^k$.
\begin{algorithmic}[1]
\Statex Let $n_i := n \eps^i$ for every $0\leq i\leq k$
\Statex Let {$M_0=M$}
\For{$i=0$ to $k-1$}
\If{$M_i \in \F^{m \times n_i}$ has $d_{M_i}(t)< \rk(M_i)-\eps n_i$}\label{alg:step_inner}
\State\label{alg:step_return} {\bf return} $M_i$
\EndIf
\State Let $k=\eps n_i = n_{i+1}$, $d_{M_i}(t) \geq \rk(M_i)-k$
\State\label{alg:step_main} By Lemma~\ref{lem:aux}, there exist $t$-row sparse $A_i\in\F^{m\times n_i}$, $M_{i+1}\in\F^{m \times n_{i+1}}$,  $B_i \in \F^{n_i\times n_i},$ and $C_i \in \F^{n_{i+1}\times n_i}$, where $M_{i+1}$ is a submatrix of $M_i$, such that:
\begin{align}\label{eq:abmcm}
A_i B_i + M_{i+1}C_i=M_i
\end{align}
\EndFor
\end{algorithmic}
\end{algorithm}

\begin{theorem}\label{thm:main}
Let $t$ and $k$ be positive integers, and let $0<\eps<1$.
If $M\in\F^{m\times n}$ is a matrix of outer dimension 
\begin{align*}
D_M\left(t k + n\eps^k\right) \geq \frac{n}{1-\eps} \, ,
\end{align*} 
then for some $n' \geq n \eps^k$, $M$ contains a submatrix $M'\in\F^{m\times n'}$  of inner dimension
\begin{align*}
d_{M'}(t)\leq \rk(M')-\eps n' \, .
\end{align*} 
Moreover, if $\F$ is a finite field of size $2^{m^{O(1)}}$, such a submatrix $M'$ can be found in time $\poly(n,m)$ with an $\mathbf{NP}$ oracle.
\end{theorem}

\begin{proof}
We prove that %
Algorithm~\ref{alg:main} must return a submatrix of $M$ of the claimed size with low inner dimension.
Assume towards a contradiction that the algorithm did not return a matrix $M_i$ in Step~\ref{alg:step_return} in any of the iterations $0\leq i \leq k-1$. Then, from Equation~\eqref{eq:abmcm}, we have
\begin{align*}
M=M_0 &=A_0 B_0 + M_1 C_0 \\
&= A_0 B_0 + (A_1 B_1 + M_2 C_1) C_0\\
&= A_0 B_0 + A_1 B_1 C_0 + M_2 C_1 C_0\\
&= A_0 B_0 + A_1 B_1 C_0 + A_2 B_2 C_1 C_0 + M_3 C_2 C_1 C_0\\
&\ldots\\
&=A_0 B_0 + A_1 B_1 C_0 + \ldots + M_{k}C_{k-1} C_{k-2}\ldots C_0 \\
&=\sum_{i=0}^{k-1} A_i B_i \prod_{j=i-1}^0 C_j + M_{k} \prod_{j=k-1}^0 C_j\\
&=
\begin{bmatrix}
A_0 & A_1 & \dots & A_{k-1} & M_{k}
\end{bmatrix}
\cdot
\begin{bmatrix}
D_0 \\
D_1 \\
\ldots\\
D_{k-1}\\
D_{k}
\end{bmatrix},
\end{align*}
where 
\begin{align*}
    D_i &:=
    \begin{cases}
      B_i \cdot \prod_{j=i-1}^0 C_j \; ,\; & \text{for}\ 1\leq i <k \\
      \prod_{j=k-1}^0 C_j, & \text{for}\  i = k \, .
    \end{cases}
\end{align*}
Now, recall that each $A_i$ is $t$-sparse, and that $M_{k} \in \F^{m\times n_{k}}$ where  $n_k=n\eps^k $. Thus, we have that
$M = AB$ where $A$ has at most $tk+n_{k}=tk+n\eps^k$ non-zero entries per row, and $B\in\F^{s \times n}$ for $s=\sum_{i=0}^{k}n_i=\sum_{i=0}^{k}n\eps^i<\frac{n}{1-\eps}$. This implies that the columns of $M$ can be generated by the columns of a $\left(tk+n\eps^k\right)$-row sparse matrix $A\in\F^{m\times n'}$, which contradicts the assumption about the outer dimension of $M$.

Now we show that one can implement Algorithm~\ref{alg:main} in time polynomial in $n$ and $m$ with an $\mathbf{NP}$ oracle. Since the language $\operatorname{INNER-DIM}$ (the language of triples $(M, t, d)$ such that $d_M(t)\geq d$) is in $\mathbf{NP}$, Step~\ref{alg:step_inner} can be done with an $\mathbf{NP}$ oracle. Step~\ref{alg:step_main} can be performed in polynomial time (with an $\mathbf{NP}$ oracle) by Lemma~\ref{lem:aux}.
\end{proof}

We conclude this section with an application of Theorem~\ref{thm:main} to data structures.
\begin{lemma}\label{lem:ds_rig}
Let $\eps>0$ be a constant. If the linear map given by a matrix $M\in\F^{m\times n}$ cannot be solved by an 
$\left(\frac{n}{1-\eps}, (t+1)\cdot \frac{\log(n/t)}{\log(1/\eps)}\right)$ linear data structure, then $M$ contains an 
$(m, n', \eps n', t)$-row rigid submatrix $M'\in\F^{m\times n'}$ for some $n'\geq t$.
\end{lemma}
\begin{proof}
Since $M$ cannot be solved by the claimed data structure, by Lemma~\ref{lem:eq}, 
\begin{align*}
D_M\left((t+1)\cdot \frac{\log(n/t)}{\log(1/\eps)}\right)> \frac{n}{1-\eps} \, .
\end{align*} 
Let us set $k=\frac{\log(n/t)}{\log(1/\eps)}$.
Then, by Theorem~\ref{thm:main}, $M$ contains a submatrix $M'\in\F^{m\times n'}$ with $d_{M'}(t)\leq \rk(M')-\eps n'$ for \begin{align*}
n'\geq n \eps^k = n \eps^{\frac{\log(n/t)}{\log(1/\eps)}} = n\cdot \frac{t}{n}=t \, .
\end{align*}
 By Lemma~\ref{lem:inner_rigidity}, %
 $M'$ is $(m, n', \eps n', t)$-row rigid.
\end{proof}

\subsection{Row Rigidity to Global Rigidity}\label{sec:global}
One drawback of Theorem~\ref{thm:main} is that the recursive algorithm produces skewed matrices 
(as we only recurse on the column space). To remedy this limitation, in this subsection 
we exhibit a reduction from worst case to average case rigidity, which will allow us to 
translate our results to square matrix rigidity with some loss in the rank parameter 
(thereby proving Theorem~\ref{thm_square_infromal}). The main ingredient of our reduction is the use of 
locally decodable codes:

\begin{definition}[Locally Decodable Codes]
A mapping $E: \F^n \mapsto \F^m$ is a $(q,\delta,\eps)$  locally decodable code (LDC) if there exists a probabilistic procedure $D\colon[n]\times \F^m\to\F$ such that
\begin{itemize}
\item For every $i\in[n]$ and $y\in\F^m$, $D(i, y)$ reads at most $q$ positions of $y$;
\item For every $i\in[n]$, $x \in\F^n$ and $v \in \F^m$ such that $|v|\leq \delta m$, 
\begin{align*}
\Pr[D(i, E(x)+v)=x_i] \geq 1-\eps \, .
\end{align*}
\end{itemize}
An LDC is called {\em linear} if the corresponding map $E$ is linear. In this case we can identify the code $E$ with its generating matrix $E \in \F^{m \times n}$.
\end{definition}

There are constructions of LDCs over all fields with $m=\poly(n), q=(\log{n})^{1+\alpha}$ for arbitrarily small $\alpha>0$, and constant $\delta$ and $\eps$ (based on Reed-Muller codes). 
\begin{lemma}[\cite{D11}, Corollary 3.14]\label{lem:ldc_constr}
Let $\F$ be a finite field. For every $\alpha, \eps > 0$ there exists $\delta=\delta(\eps)>0$ and an explicit family of $((\log{n})^{1+\alpha}, \delta, \eps)$-linear LDCs $M\in\F^{m\times n}$ for $m=n^{O(1/\alpha)}$.
\end{lemma}

We will use the following property of linear LDCs.
\begin{lemma}[Implicit in~\cite{GKST02,DS07}]\label{lem:ldc}
Let $E\in\F^{m\times n}$ be a $(q,\delta,3/4)$ linear LDC, and let $R$ be a set of at least $(1-\delta)m$ rows of $E$. For any $i\in[n]$, there exists a set of $q$ rows in $R$ which spans the $i$th standard basis vector $e_i$.
\end{lemma}

We are now ready to present the main result of this section.
\begin{theorem}\label{thm:global}
Let $M\in\F^{m\times n}$ be a matrix,  $E\in\F^{m'\times m}$ be a $(q,\delta,3/4)$-linear LDC, and let $A=EM$ (i.e., the matrix obtained by applying $E$ to each column of $M$). 
\begin{itemize}
\item If $M$ is $(m,n,r,t+1)$-row rigid, then $A$ is $(m',n,r, \frac{\delta tm'}{q})$-globally rigid.
\item If $M$ is $(m,n,r,t+1)$-strongly row rigid, then $A$ is $(m',n,r, \frac{\delta tm'}{q})$-strongly globally rigid.
\end{itemize}
\end{theorem}

\begin{proof}
\leavevmode
\begin{itemize}
\item
Assume towards a contradiction that $A$ is not $(m',n,r, \frac{\delta tm'}{q})$-globally rigid. Then $A=L+S$, where $\rk(L)\leq r$ and $S$ is $ \frac{\delta tm'}{q}$-globally sparse. Let $S'$ be the set of rows of $S$ with at most $\frac{t}{q}$ non-zero elements. By Markov's inequality there are at least $(1-\delta)m'$ rows in $S'$, let $L'$ be the corresponding rows of $L$. By row-rigidity of $M$, some row $i$ of $M$ is $(t+1)$-far from the space generated by the rows of $L$ (that it, the Hamming distance between this row and any vector in the span of the rows of $L$ is at least $t+1$). 
By Lemma~\ref{lem:ldc}, there are $q$ rows in $L'+S'$ which span the $i$th row of $M$. In particular, the $i$th row of $M$ has distance at most $q\cdot \frac tq$ from the rowspan of $L'$, which contradicts the assumption that this row is $(t+1)$-far from the space generated by the rows of $L$.
\item In order to show that the resulting matrix $A$ is strongly rigid, it suffices to show that the application of linear LDC commutes with basis changes. Assume towards a contradiction that the resulting matrix $A$ is not strongly globally rigid. This implies that there exists an invertible matrix $U\in\F^{n\times n}$ such that 
\begin{align*}
AU=(E M) U = E (M U) \,
\end{align*}
is not $(m',n,r, \frac{\delta tm'}{q})$-globally rigid. Notice that from strong rigidity of $M$, we have that $MU$ is $(m,n,r,t+1)$-row rigid. Thus, by the first item of this theorem, $AU$ is also $(m',n,r, \frac{\delta tm'}{q})$-globally rigid.
\end{itemize}
\end{proof}

We remark that the same argument as in the proof of Theorem~\ref{thm:global} can be also used to give a worst-case to average-case reduction for linear \emph{data-structures} (with a similar loss of $(\log{n})^{1+\alpha}$ factor in the number of probes).

We next show that given a rectangular $ m \times n$ matrix  and \emph{row rigidity} $t$, one can efficiently produce a square matrix of size $m' \times m'$ for $m'=n^{O(1/\alpha)}$ with row rigidity $\frac{m'}{n} \cdot \frac{t}{(\log{n})^{1+\alpha}}$ (for the same rank parameter). That is, one can increase the rigidity proportionally to the increase in size with a loss of only $(\log{n})^{1+\alpha}$ factor.
\begin{corollary}\label{cor:square}
For every constant $\alpha>0$, there is a polynomial-time algorithm which given an $(m,n,r,t+1)$-row rigid matrix $M\in\F^{m\times n}$, outputs a \emph{square matrix} $A\in\F^{m'\times m'}$ which is $\left(m',m',r,\frac{m'^2}{n}\cdot\frac{ t}{(\log{m})^{1+\alpha}}\right)$-globally rigid for $m'=m^{O(1/\alpha)}$. In particular, $A$ is $\left(m',m',r,\frac{m'}{n}\cdot\frac{t}{(\log{m})^{1+\alpha}}\right)$-row rigid.
\end{corollary}

\begin{proof}
Let $E\in\F^{m'\times m}$ be a $((\log{m})^{1+\alpha/2},\delta,\frac{3}{4})$-linear LDC (whose efficient construction is guaranteed by Lemma~\ref{lem:ldc_constr}) for constant $\delta$, and let $m'=m^{O(1/\alpha)}$ be a multiple of $n$. Then we construct $A\in\F^{m'\times m'}$ by stacking side by side $(m'/n)$ copies of $EM$.
 
By Theorem~\ref{thm:global}, $EM$ is $(m',n,r, \frac{tm'}{(\log{m})^{1+\alpha}})$-globally rigid. In order to reduce the rank of $A$ to $r$, one needs to reduce the rank of each copy of $EM$ to at most $r$. Therefore, one needs to change at least $\frac{tm'}{(\log{m})^{1+\alpha}}\cdot \frac{m'}{n}=\frac{m'^2}{n}\cdot\frac{ t}{(\log{m})^{1+\alpha}}$ entries in $A$ in order to get rank at most $r$. This implies that $A$ has global rigidity $\frac{m'^2}{n}\cdot\frac{ t}{(\log{m})^{1+\alpha}}$ and row rigidity $\frac{m'}{n}\cdot\frac{ t}{(\log{m})^{1+\alpha}}$ for the rank parameter $r$.
\end{proof}

\section{Data Structures and Rigidity} \label{sec_DS_rigidity}
In Section~\ref{sec:prelim} we showed that a strong upper bound on the \emph{inner dimension} of a matrix implies that the matrix has non-trivial rigidity, and that a strong lower bound on the \emph{outer dimension} implies that the corresponding linear transformation cannot be computed by an efficient linear data structure. In this section we 
use the relations between inner and outer dimensions, to show that any improvement on rigidity lower bounds will lead to higher 
data structures lower bounds (against linear space), while improvements (on \eqref{eq_cell_sampling_LB}) 
in data structure lower bounds would yield new rigidity lower bounds. We state these (different) implications in various space regimes.

\subsection{Linear Space}
We will make use of the following known relation between inner and outer dimensions.
\begin{restatable}[\cite{PP06}]{proposition}{ppdims}
\label{prop:ppdims}
$
d_V(t)+D_V(t) \geq 2\dim{V} .
$
\end{restatable}

This proposition directly yields the following connection between rigidity and data structure lower bounds.
\begin{corollary}\label{cor:rig_to_ds}
If a matrix $M\in\F^{m\times n}$ is $(m,n, r, t)$-strongly row rigid, then the corresponding linear map cannot be computed by an $(n + r-1, t)$ linear data structure.
\end{corollary}
We remark that this corollary works for any function $r=r(n)$, including the regimes where $r=O(n)$ and $r=o(n)$.
\begin{proof}
From Lemma~\ref{lem:inner_rigidity}, we have that $d_M(t)\leq n-r$. By Proposition~\ref{prop:ppdims}, this gives us that $D_M(t)\geq n+r$. Now Lemma~\ref{lem:eq} gives us that no  $(n+r-1, t)$ linear data structure can compute $M$.
\end{proof}
In particular, an $(m,n,(1+\eps)n, t)$-strongly row rigid matrix implies a lower bound of $t$ on the query time of linear data structures with linear space $s=(1+\eps)n$.
We remark that the best known rigidity bound in this regime for $m=\poly(n)$ is $t=\Omega(\log{n})$ which matches the best known lower bound for linear space data structures. 
Any improvement to $t=\omega(\log{n})$ on the known rigidity construction would lead to a new data structure lower bound (against data structures with small linear space).

Now we use Lemma~\ref{lem:ds_rig} to show that the opposite direction (with a slight change of parameters) also holds for a submatrix of $M$.
\begin{theorem}\label{thm:dsbound}
\leavevmode
\begin{enumerate}
\item (Poly-logarithmic Lower Bounds) Let $\eps>0$ and $c\geq 1$ be constants. If the linear map given by a matrix $M\in\F^{m\times n}$ cannot be solved by an $\left(\frac{n}{1-\eps}, (\log{n})^{c} \right)$ linear data structure, then $M$ contains an $(m, n', \eps n', \alpha\cdot(\log{n})^{c-1})$-row rigid submatrix $M'\in\F^{m\times n'}$ for some constant $\alpha>0$ and $n'\geq \alpha\cdot(\log{n})^{c-1}$.

\item (Polynomial Lower Bounds) Let $\eps, \delta>0$ be constants. If the linear map given by a matrix $M\in\F^{m\times n}$ cannot be solved by an $\left(\frac{n}{1-\eps}, n^{\delta} \right)$ linear data structure, then $M$ contains an $(m, n', \eps n', n^{\alpha})$-row rigid submatrix $M'\in\F^{m\times n'}$ for any $\alpha<\delta$ and some $n'\geq n^{\alpha}$.

\item (Square Rigidity) Let $\eps>0,\gamma>0$ and $c>2$ be constants. If the linear map given by a matrix $M\in\F^{m\times n}$ cannot be solved by an $\left(\frac{n}{1-\eps}, (\log{n})^{c} \right)$ linear data structure, then there is a \emph{square} matrix $M'\in\F^{m'\times m'}$ for $m'=m^{O(1)}$, such that $M'$ is $(m', m', r, \frac{m' (\log{n})^{c-2-\gamma}}{r})$-row rigid for some $r$ 
(which depends on $n$).
\end{enumerate}
Moreover, if $|\F|=2^{m^{O(1)}}$ and $M\in{\bf P^{NP}}$, then the family of matrices $M'$ of high rigidity belongs to the class $\mathbf{P^{NP}}=\mathbf{DTIME}[\operatorname{poly}(m)]^{\mathbf{NP}}$.\footnote{When we say that a matrix $M$ belongs to ${\bf P^{NP}}$, we mean that there exists a family of matrices $M_n\in\F^{m(n)\times n}$ for infinitely many values of $n$ such that each $M_n$ can be computed by a polynomial time algorithm with an $\mathbf{NP}$ oracle.}
\end{theorem}
\begin{proof}
\leavevmode
\begin{enumerate}
\item Let us set $t=\frac{(\log{n})^{c-1}}{\log{1/\eps}}-1$. Now have that $M$ cannot be solved by an $\left(\frac{n}{1-\eps}, (t+1)\cdot \frac{\log(n/t)}{\log(1/\eps)}\right)$ linear data structure. This, together with Lemma~\ref{lem:ds_rig}, implies that $M$ contains an $(m, n', \eps n', \frac{(\log{n})^{c-1}}{\log{1/\eps}}-1)$-row rigid submatrix.
 
 \item Here we set $t=\frac{n^\delta}{\log(1/\eps)\log(n)}-1$. Again, Lemma~\ref{lem:ds_rig} implies that $M$ contains an $(m, n', \eps n', t)$-row rigid submatrix for $n'\geq t$.
 
 \item From the first bullet of this theorem, we get an $(m, n', \eps n', \alpha\cdot(\log{n})^{c-1})$-rigid submatrix $M'\in\F^{m\times n'}$. Now we apply Corollary~\ref{cor:square} to get an $m' \times m'$ matrix which has row rigidity $m' \cdot\frac{(\log{n})^{c-2-\gamma}}{n'}$ for the rank parameter $r=\eps n'$.
\end{enumerate}
\end{proof}
We note that a data structure lower bound of $t\geq\omega\left((\log{n})^2\right)$ will lead to a new bound on rigidity of rectangular matrices. Moreover, by the last bullet of this theorem, we have that a lower bound of of $t\geq\Omega\left((\log{n})^{3+\eps}\right)$ will lead to a new bound on rigidity of \emph{square} matrices: it will give us a matrix which is $(n,n,r,s)$-row rigid for $s\geq \Omega\left(\frac{n}{r}\cdot\frac{t}{(\log{n})^{2+\eps/2}}\right)$ (which is better than the known bound of $s\geq\Omega\left(\frac{n}{r}\log\frac{n}{r}\right)$).

\subsection{Super-linear Space}
Recall that any data structure problem has two trivial solutions: $s=n, t=n$, and $s=m, t=1$. A simple counting argument (see Lemma~\ref{lem:counting} in Appendix~\ref{apx:a}) 
shows that for any $s<0.99m$, a random linear problem requires $t=\Omega(n/\log{s})$ query time. Here we show that near-optimal data structure lower 
bound, against space $s\geq\omega(m/\log\log{m})$, would imply a super-linear lower bound against log-depth circuits. 

\begin{theorem}\label{thm:clb}
Let $M\in\F^{m\times n}$ be a matrix for $m\geq n$. If for some constant $\eps>0$ and every constant $\delta>0$
\begin{itemize}
\item $M$ cannot be computed by $\left(\frac{\delta m}{\log\log{m}}+n,m^{\eps}\right)$ linear data structures,
\item or $M$ cannot be computed by $\left(\eps m+n, 2^{(\log{m})^{1-\delta}}\right)$ linear data structures,
\end{itemize}
then $M$ cannot be computed by linear circuits of size $O(m)$ and depth $O(\log{m})$.
\end{theorem}
\begin{proof}
Assume towards a contradiction that $M$ can be computed by a circuit of size $cm$ and depth $c\log{m}$ for a constant $c$. Then, by Theorem~\ref{thm:valiant}, $M=A+C\cdot D$, where  $A\in\F^{m\times n}, C\in\F^{m\times s}, D\in\F^{s\times n}$, $A$ and $C$ are $t$-row sparse. In particular, the column space of $M$ is spanned by the column spaces of $t$-row sparse matrices $A$ and $C$. That is, $D_M(t)\leq n+s$. By Lemma~\ref{lem:eq}, $M$ can be computed by an $(n+s, t)$ linear data structure.
\end{proof}
While the bounds given in this theorem are interesting in the regime $m,s\gg n$ , they also give a curious corollary for $m=O(n)$. For example, in the regime of $m=O(n)$ for $s=\frac{\delta m}{\log\log{m}}+n$ we know a lower bound of $t\geq \Omega(\log{n})$~\cite{Lar12}. An improvement of this bound to sub-polynomial $2^{(\log{n})^{1-\delta}}$ would give a super-linear circuit lower bound.
\begin{corollary}
Let $m=O(n)$. If for some constant $\eps>0$ and every constant $\delta>0$ a linear map $M\in\F^{m\times n}$ cannot be solved by $\left(n(1+\eps), 2^{(\log{n})^{1-\delta}} \right)$ linear data structures, then $M$ cannot be computed by linear circuits of size $O(m)$ and depth $O(\log{m})$.
\end{corollary}

\subsection{Succinct Space}
In the succinct regime, data structures can only use space $s=n+o(n)$. In this regime we know strong lower bounds for data structures. Namely, if $s=n+r$ for $r=o(n)$, then the best know lower bound over $\F_2$ is $t\geq \frac{n}{r}$~\cite{gal:succinct}. We will show that the succinct case corresponds exactly to the case of strong rigidity in the regime $r=o(n)$, and will use this connection to improve the known data structure lower bound by a logarithmic factor. We remark that one can extract the same lower bound of $t\geq \frac{n\log{n}}{r}$ from~\cite{Lar12} for a problem with \emph{polynomially} many queries $m=\Omega(n^{1+\eps})$, while our simple construction gives it for linear number of queries $m=O(n)$.

\begin{theorem}
\label{thm:succinct}
Let $1\leq r(n)\leq n^{1-\eps}$ and $(\log{n})^{\delta}\leq t(n)\leq n$ be non-decreasing and time-constructible functions for some constant $\eps,\delta>0$. Then\footnote{This lemma applies to the whole spectrum of $r=o(n)$, but for the ease of presentation we restrict our attention to the regime of $r\leq n^{1-\eps}$.}
\begin{itemize}
\item An $(m,n,r(n),t(n))$-strongly row rigid matrix $M\in\F^{m\times n}$ cannot be computed by $(n+r(n)-1, t(n))$ linear data structures.
\item If there exists a constant $\mu>0$, such that $M\in\F^{m\times n}$ cannot be computed by $\left(n+(1+\mu)r(n), (1+\mu)t(n)\right)$ linear data structure and $M\in {\bf P^{NP}}$, then there is $(m,n',r(n'),t(n'))$-strongly row rigid $M'\in \mathbf{P^{NP}}=\mathbf{DTIME}[\operatorname{poly}(m)]^{\mathbf{NP}}$ for some $n'\geq \mu t(n)/2$.
\end{itemize}
\end{theorem}

\begin{proof}
The first item of the Theorem follows directly from Corollary~\ref{cor:rig_to_ds}. For the other direction, we will run Algorithm~\ref{alg:main} with slightly modified parameters.

For a positive integer $i$, let $r^{(i)}$ denote the composition of $r$ with itself $i$ times:
\begin{align*}
r^{(i)}(n)=\underbrace{r \circ \dots \circ r}_{i\:\text{times}}(n)\,.
\end{align*}
Let $k$ be the smallest number such that $n^{\eps^{k}}\leq \mu t(n)/2$. 
We define $n_0=n$, then $n_{i}=r(n_{i-1})$ for $0< i<k$, and $n_k=\mu t(n)/2$. For $0\leq i\leq k$, we define $t_i=t(n_i)$ and $r_i=r(n_i)$. Let us now run Algorithm~\ref{alg:main}. In Step~\ref{alg:step_inner}, the algorithm will check whether $d_{M_i}(t_i)<\rk(M_i)-r_i$. If this inequality is satisfied, then the algorithm returns an $(m,n',r(n'),t(n'))$-strongly row rigid matrix $M'$. Again, as in the proof of Theorem~\ref{thm:main}, this algorithm can be implemented in polynomial time with an $\mathbf{NP}$ oracle.

If the algorithm does not return a matrix $M_i$ for any $0\leq i\leq k-1$, then we get a factorization $M=AB$, where the matrix $A$ has at most $t'=\sum_{i=0}^{k-1}t_k +n_k$ non-zero entries per row, and $B\in\F^{s\times n}$ for $s'=\sum_{i=0}^k n_i$. In particular, $M$ can be computed by an $(s',t')$ linear data structure. 

From $n_i=r(n_{i-1})\leq (n_{i-1})^{1-\eps}$, we have $s'\leq n + (1+\mu)r(n)$ for large enough values of $n$. Now, from  $r(n)\leq n^{1-\eps}$ and $t(n)\geq (\log{n})^{\delta}$, we have  
\begin{align*}
t_i=t(n_i)=t(r(n_{i-1})) \leq  (1-\eps)^{\delta}\cdot t(n_{i-1}) = (1-\eps)^{\delta}\cdot t_{i-1}\, .
\end{align*}
Thus, 
\begin{align*}
t'=\sum_{i=0}^{k-1}t_k +n_k
\leq t(n) \cdot\frac{1}{1-(1-\eps)^{\delta}}+\mu t(n)/2
\leq (1+\mu)t(n)
\end{align*}
for $\mu = \frac{2}{1-(1-\eps)^{\delta}}$. This contradicts the assumption that $M$ cannot be computed by $(n+(1+\mu)r(n), (1+\mu)t(n))$ linear data structure.
\end{proof}

We remark that one can also apply Corollary~\ref{cor:square} here to obtain square rigid matrices (see, e.g., Corollary~\ref{cor:raz}).

We claim that the best known rigidity lower bound $t=\Omega\left(\frac{n}{r}\log{\frac{n}{r}}\right)$ gives us the same lower bound on strong rigidity, and, thus, improves the known succinct data structures lower bounds by a logarithmic factor.

In the following we will make use of error-correcting codes with constant rate and constant relative distance (see, e.g., Justesen and Goppa codes).
\begin{proposition}[\cite{macwilliams1977theory,lint1988introduction,van2012introduction}]
For any finite field $\F$ there exists an explicit family of linear error correcting codes with rate $1/4$ and constant relative distance $\delta=\Theta(1)$.
\end{proposition}

In Appendix~\ref{apx:a} we modily the proof of Friedman~\cite{friedman1993note} to get strong rigidity.
\begin{restatable}[\cite{friedman1993note}]{proposition}{strong}
\label{prop:strong}
Let $\F$ be a finite filed of size $|\F|=q$, and let $M\in\F^{n\times n/4}$ be a transposed generator matrix of a code with constant relative distance $\delta$. That is, the columns $c_1,\ldots,c_{n/4}\in\F^n$ of $M$ form a basis of a linear code. Then $M$ is $(n,n/4,r,t)$-strongly row rigid for any $r\geq \log{n}$ and any
\begin{align*}
1\leq t\leq O\left(\frac{n}{r}\left( \log_q\left({\frac{n}{r}}\right) +\log_q(q-1)\right)  \right) = O\left(\frac{n}{r}\cdot\max\left( \log_q\left({\frac{n}{r}}\right),1\right)  \right)\, .
\end{align*}
\end{restatable}

As a corollary of Theorem~\ref{thm:succinct} and Proposition~\ref{prop:strong}, we get a new data structure lower bound for the succinct case.
\begin{corollary}
Let $\F$ be a finite filed of size $|\F|=q$, and let $M\in\F^{4n\times n}$ be a transposed generator matrix of a code with constant relative distance. Then for any $\log{n}\leq r\leq n^{1-\eps}$ and any $1\leq t\leq O\left(\frac{n}{r}\left( \log_q\left({\frac{n}{r}}\right) +\log_q(q-1)\right)  \right)$, $M$ cannot be computed by a linear $(n+r, t)$ data structure for large enough $n$.
\end{corollary}

We also note that improving this data structure lower bound for $s=n+2^{(\log\log{n})^{\omega(1)}}$ would resolve a big open problem in communication complexity.
\begin{proposition}[\cite{R89,Wun12}]\label{prop:raz}
If $M\in\F^{n\times n}$ is $(r,t)$-row rigid for $r=2^{(\log\log{n})^{\omega(1)}}$ and $t=n/2^{(\log\log{n})^{O(1)}}$, then the language $L$ corresponding to $M$ is not in the polynomial hierarchy for communication complexity $L\not\in{\mathbf{ PH}}^{cc}$.
\end{proposition}
Theorem~\ref{thm:succinct}, Proposition~\ref{prop:raz}, and Corollary~\ref{cor:square} give us the following result.
\begin{corollary}\label{cor:raz}
If $M\in\F^{m\times n}$ cannot be computed by $(n+r,t)$ linear data structures for $m=n^{O(1)}$, $r=2^{(\log\log{n})^{\omega(1)}}$ and $t=n/2^{(\log\log{n})^{O(1)}}$, and $M\in\mathbf{P^{NP}}$, then there exists a language $L\in\mathbf{P^{NP}}$ such that $L\not\in{\mathbf{ PH}}^{cc}$.
\end{corollary}
\begin{proof}
By Theorem~\ref{thm:succinct}, $M$ contains an $(m,n',r(n'),t(n'))$-row rigid submatrix $M'$. Now Corollary~\ref{cor:square} gives us a matrix $A$ which is $(m',m',r(n'),m'\frac{t}{n((\log{m})^{1+\alpha})})$-row rigid. From $m'=n^{O(1)}$, we have that $A$ is $(m',m',2^{(\log\log{m'})^{\omega(1)}},m'/2^{(\log\log{m'})^{O(1)}})$-row rigid, which finishes the proof.
\end{proof}

\bibliographystyle{alpha}
\bibliography{refs}

\begin{thebibliography}{MNSW98}

\bibitem[AB09]{AB2009}
Sanjeev Arora and Boaz Barak.
\newblock {\em Computational complexity: a modern approach}.
\newblock Cambridge University Press, 2009.

\bibitem[AE99]{AE99}
Pankaj~K. Agarwal and Jeff Erickson.
\newblock Geometric range searching and its relatives.
\newblock {\em Contemp. Math.}, 223:1--56, 1999.

\bibitem[Aga04]{Agarwal04}
Pankaj~K. Agarwal.
\newblock Range searching.
\newblock In {\em Handbook of Discrete and Computational Geometry, Second
  Edition.}, pages 809--837. Chapman and Hall/CRC, 2004.

\bibitem[APY09]{APY09}
Noga Alon, Rina Panigrahy, and Sergey Yekhanin.
\newblock Deterministic approximation algorithms for the nearest codeword
  problem.
\newblock In {\em RANDOM 2009}, pages 339--351, 2009.

\bibitem[BL13]{BL13}
Karl Bringmann and Kasper~Green Larsen.
\newblock Succinct sampling from discrete distributions.
\newblock In {\em STOC 2013}, pages 775--782, 2013.

\bibitem[BL15]{BL15}
Joshua Brody and Kasper~Green Larsen.
\newblock Adapt or die: Polynomial lower bounds for non-adaptive dynamic data
  structures.
\newblock {\em Theory Comput.}, 11:471--489, 2015.

\bibitem[Cha90]{Chazelle90}
Bernard Chazelle.
\newblock Lower bounds for orthogonal range searching: Part {II}. {T}he
  arithmetic model.
\newblock {\em J. ACM}, 37(3):439--463, 1990.

\bibitem[DS07]{DS07}
Zeev Dvir and Amir Shpilka.
\newblock Locally decodable codes with two queries and polynomial identity
  testing for depth 3 circuits.
\newblock {\em SIAM J. Comput.}, 36(5):1404--1434, 2007.

\bibitem[Dvi11]{D11}
Zeev Dvir.
\newblock On matrix rigidity and locally self-correctable codes.
\newblock {\em Comput. Complexity}, 20(2):367--388, 2011.

\bibitem[EGS75]{erdos1975sparse}
Paul Erd{\"o}s, Ronald~L. Graham, and Endre Szemer{\'e}di.
\newblock On sparse graphs with dense long paths.
\newblock {\em Comp. and Math. with Appl}, 1:145--161, 1975.

\bibitem[Fre81]{Fred81}
Michael~L. Fredman.
\newblock A lower bound on the complexity of orthogonal range queries.
\newblock {\em J. ACM}, 28(4):696--705, 1981.

\bibitem[Fri93]{friedman1993note}
Joel Friedman.
\newblock A note on matrix rigidity.
\newblock {\em Combinatorica}, 13(2):235--239, 1993.

\bibitem[GKST02]{GKST02}
Oded Goldreich, Howard Karloff, Leonard~J. Schulman, and Luca Trevisan.
\newblock Lower bounds for linear locally decodable codes and private
  information retrieval.
\newblock In {\em CCC 2002}, pages 175--183, 2002.

\bibitem[GM07]{gal:succinct}
Anna G\'{a}l and Peter~Bro Miltersen.
\newblock The cell probe complexity of succinct data structures.
\newblock {\em Theor. Comput. Sci.}, 379:405--417, July 2007.

\bibitem[GT16]{goldreich2016matrix}
Oded Goldreich and Avishay Tal.
\newblock Matrix rigidity of random toeplitz matrices.
\newblock In {\em Proceedings of the forty-eighth annual ACM symposium on
  Theory of Computing}, pages 91--104. ACM, 2016.

\bibitem[IW01]{IW01}
Russell Impagliazzo and Avi Wigderson.
\newblock Randomness vs time: Derandomization under a uniform assumption.
\newblock {\em J. Comput. Syst. Sci.}, 63(4):672--688, 2001.

\bibitem[JS11]{JS11}
Stasys Jukna and Georg Schnitger.
\newblock Min-rank conjecture for log-depth circuits.
\newblock {\em J. Comput. Syst. Sci}, 77(6):1023--1038, 2011.

\bibitem[Juk12]{J12}
Stasys Jukna.
\newblock {\em Boolean function complexity: advances and frontiers}, volume~27.
\newblock Springer Science \& Business Media, 2012.

\bibitem[KU11]{KU08}
Kiran~S. Kedlaya and Christopher Umans.
\newblock Fast polynomial factorization and modular composition.
\newblock {\em {SIAM} J. Comput.}, 40(6):1767--1802, 2011.

\bibitem[KV18]{KV18}
Mrinal Kumar and Ben~Lee Volk.
\newblock Rigid matrices: beyond the untouched minor argument in subexponential
  time.
\newblock {\em Manuscript}, 2018.

\bibitem[Lar12]{Lar12}
Kasper~Green Larsen.
\newblock Higher cell probe lower bounds for evaluating polynomials.
\newblock In {\em FOCS 2012}, pages 293--301, 2012.

\bibitem[Lar14]{Lar14}
Kasper~Green Larsen.
\newblock On range searching in the group model and combinatorial discrepancy.
\newblock {\em {SIAM} J. Comput.}, 43(2):673--686, 2014.

\bibitem[LG88]{lint1988introduction}
Jacobus Hendricus~van Lint and Gerard van~der Geer.
\newblock {\em Introduction to coding theory and algebraic geometry}.
\newblock Birkh{\"a}user Basel, 1988.

\bibitem[Lok09]{Lokam09}
Satyanarayana~V. Lokam.
\newblock Complexity lower bounds using linear algebra.
\newblock {\em Found. Trends Theor. Comput. Sci.}, 4(1-2):1--155, 2009.

\bibitem[Mil93]{Milt93}
Peter~Bro Miltersen.
\newblock The bit probe complexity measure revisited.
\newblock In {\em STACS 1993}, pages 662--671, 1993.

\bibitem[MNSW98]{MNSW98}
Peter~Bro Miltersen, Noam Nisan, Shmuel Safra, and Avi Wigderson.
\newblock On data structures and asymmetric communication complexity.
\newblock {\em J. Comput. Syst. Sci.}, 57(1):37--49, 1998.

\bibitem[MS77]{macwilliams1977theory}
Florence~Jessie MacWilliams and Neil James~Alexander Sloane.
\newblock {\em The theory of error-correcting codes}.
\newblock Elsevier, 1977.

\bibitem[P{\v a}t07]{PatGroup07}
Mihai P{\v a}tra{\c s}cu.
\newblock Lower bounds for 2-dimensional range counting.
\newblock In {\em STOC 2007}, pages 40--46, 2007.

\bibitem[P{\v a}t08]{patrascu08structures}
Mihai P{\v a}tra{\c s}cu.
\newblock Unifying the landscape of cell-probe lower bounds.
\newblock In {\em FOCS 2008}, pages 434--443, 2008.

\bibitem[PP06]{PP06}
Ramamohan Paturi and Pavel Pudl{\'a}k.
\newblock Circuit lower bounds and linear codes.
\newblock {\em J. Math. Sci.}, 134(5):2425--2434, 2006.

\bibitem[PR94]{pudlak1994some}
Pavel Pudl{\'a}k and Vojtech R{\"{o}}dl.
\newblock Some combinatorial-algebraic problems from complexity theory.
\newblock {\em Discrete Math.}, 136(1-3):253--279, 1994.

\bibitem[PTW10]{PTW10}
Rina Panigrahy, Kunal Talwar, and Udi Wieder.
\newblock Lower bounds on near neighbor search via metric expansion.
\newblock In {\em FOCS 2010}, pages 805--814, 2010.

\bibitem[Raz89]{R89}
Alexander~A. Razborov.
\newblock On rigid matrices.
\newblock {\em Manuscript}, 1989.
\newblock In Russian.

\bibitem[Sie04]{Siegel04}
Alan Siegel.
\newblock On universal classes of extremely random constant-time hash
  functions.
\newblock {\em SIAM J. Comput.}, 33(3):505--543, 2004.

\bibitem[SSS97]{shokrollahi1997remark}
Mohammad~Amin Shokrollahi, Daniel~A. Spielman, and Volker Stemann.
\newblock A remark on matrix rigidity.
\newblock {\em Inf. Process. Lett.}, 64(6):283--285, 1997.

\bibitem[SY11]{SY11}
Shubhangi Saraf and Sergey Yekhanin.
\newblock Noisy interpolation of sparse polynomials, and applications.
\newblock In {\em CCC 2011}, pages 86--92, 2011.

\bibitem[Val77]{Valiant}
Leslie~G. Valiant.
\newblock Graph-theoretic arguments in low-level complexity.
\newblock In {\em MFCS 1977}, pages 162--176, 1977.

\bibitem[vEB90]{BoasPointerMachine90}
Peter van Emde~Boas.
\newblock Machine models and simulation.
\newblock In {\em Handbook of Theoretical Computer Science, Volume {A:}
  Algorithms and Complexity}, pages 1--66. MIT Press, 1990.

\bibitem[Vio18]{V18}
Emanuele Viola.
\newblock Lower bounds for data structures with space close to maximum imply
  circuit lower bounds.
\newblock In {\em ECCC}, volume~25, 2018.

\bibitem[vL12]{van2012introduction}
Jacobus~Hendricus van Lint.
\newblock {\em Introduction to coding theory}, volume~86.
\newblock Springer Science \& Business Media, 2012.

\bibitem[Wun12]{Wun12}
Henning Wunderlich.
\newblock On a theorem of {R}azborov.
\newblock {\em Comput. Complex.}, 21(3):431--477, 2012.

\bibitem[Yao81]{Yao81}
Andrew Chi-Chih Yao.
\newblock Should tables be sorted?
\newblock {\em J. ACM}, 28(3):615--628, July 1981.

\end{thebibliography}

\appendix
\section{Proofs of Some Known Results}\label{apx:a}
In this appendix we give proofs of some known statements adjusted to our definitions. Propositions~\ref{prop:ppdims} and~\ref{prop:strong} were proven in~\cite{PP06} and~\cite{friedman1993note} for definition concerning column sparsity rather than row sparsity (which matters in this context), and Proposition~\ref{prop:linearization} was proven in \cite{JS11} for linear depth-$2$ circuits rather than for linear data structures. Lemma~\ref{lem:counting} gives a simple upper bound for all linear data structure problems, and a matching non-constructive lower bound. Lemma~\ref{lem:simplelemma} gives an alternative proof of Lemma~\ref{lem:inner_rigidity}.
\linearization*
\begin{proof}
Let $P\colon\F^n\to\F^s$ be the preprocessing function of $\cD$, and let $Q\in\F^{m\times s}$ be a linear transformation computed by the query function of $\cD$. Let $\bm{e_1},\ldots,\bm{e_n}$ be the unit vectors in $\F^n$. Consider a new linear data structure $\cD'$ where the preprocessing function computes a linear transformation, which for a vector $\bm{x}=\sum_{i=1}^n x_i\bm{e_i}$ outputs $P'(x)=\sum_{i=1}^n x_i P(\bm{e_i})$, and the query function stays the same: $Q'(x)=Q(x)$.

Since the original data structure $\cD$ computes the linear transformation $V$, it holds that:
\begin{align*}
\forall x\in\F^n\colon V\bm{x} = Q\cdot P(x) \, .
\end{align*}
Now, by the linearity of matrix-vector products, the new data structure computes 
\begin{align*}
Q'\cdot P'(x)
= Q\cdot\left(\sum_{i=1}^n x_i P(\bm{e_i}) \right)
= \sum_{i=1}^n x_i Q\cdot P(\bm{e_i})
= \sum_{i=1}^n x_i V\bm{e_i}
= V \cdot \left(\sum_{i=1}^n x_i \bm{e_i} \right)
= V\bm{x}
\, .
\end{align*}
\end{proof}
\ppdims*
\begin{proof}
Let $U\supseteq V$ be a $t$-sparse subspace of dimension $D_V(t)$ (the existence of $U$ is guaranteed by the definition of outer dimension). Let $A_U$ be a $t$-sparse matrix generating $U$, and let $A_W$ be the first $\dim(V)$ columns of $A_U$. Now let $W$ be the column space of $A_W$. Clearly $A_W$ and $W$ are also $t$-sparse. From the definition of the inner dimension, we have that $d_V(t)\geq \dim(V \cap W)$. 

On the other hand, $U$ contains $V$ and $W$. Thus,
\begin{align*}
D_V(t)=\dim(U)\geq\dim(V + W)=\dim(V)+\dim(W)-\dim(V \cap W) \geq 2\dim(V)-d_V(t) \, .
\end{align*}
\end{proof}

\strong*
\begin{proof}
By Lemma~\ref{lem:inner_rigidity}, it suffices to show that $d_M(t)\leq n/4 - r$. Let $B\in\F^{n \times n/4}$ be a $t$-sparse matrix, and let $C=(c_1,\ldots,c_{n/8})\in\F^{n\times n/8}$ be the $n/8$ sparsest columns of $B$. Note that by Markov's inequality, each column of $C$ has at most $8t$ non-zero entries. Let $V$ be the column space of $M$, and $U$ be the column space of $C$. We will show that $\dim(U\cap V)\leq n/8-r$, which will finish the proof.

Let 
\begin{align*}
W=\left\{ w=(w_1,\ldots,w_{n/8})\in\F^{n/8}\colon \sum_{i=1}^{n/8} w_i c_i \in U \right\} \,.
\end{align*}
Assume towards a contradiction that $\dim(W)=\dim(U\cap V)> n/8-r$. This implies that $W$ contains a non-zero point of Hamming weight at most $a$ for any $a$ such that
\begin{align*}
 \left|\text{Hamming ball of radius }a/2\text{ in } \F^{n/8}\right|\geq q^r\,.
\end{align*}
In particular, there is a point of Hamming weight at most $a$ for $a$ satisfying
\begin{align*}
{n/8 \choose a/2}  \left( q-1\right)^{a/2} \geq q^r\,.
\end{align*} 
On the other hand, a point of Hamming weight $a$ in $W$, gives a non-zero codeword of Hamming weight at most $a\cdot 8k$. Since we know that all non-zero codewords have Hamming weight at least $\delta n$, we get $a\geq\frac{\delta n}{8k}$. Now we have
\begin{align*}
r&\geq \log_q\left({n/8 \choose a/2}  \left( q-1\right)^{a/2} \right)\\
&\geq \frac{a}{2}\log_q\left(\frac{n}{4a}\right)+\frac{a}{2}\log_q(q-1)\\
&\geq\frac{\delta n}{16t}\left(\log_q\left(\frac{2t}{\delta}\right)+\log_q(q-1)\right)\\
&=\Omega\left(\frac{n}{t}\left(\log_q(t)+\log_q(q-1)\right)\right) \,.
\end{align*}
Or, equivalently, $t\geq\Omega\left(\frac{n}{r}\left( \log_q\left({\frac{n}{r}}\right) +\log_q(q-1)\right)  \right)$, which leads to a contradiction.
\end{proof}

\begin{lemma}\label{lem:counting}
Let $\F$ be a finite field of size $|\F|=q$, and let $\eps>0$ be a constant.
\begin{itemize}
\item  For any $n\leq s\leq (1-\eps)m$, there exists a linear datat structure problem $M\in\F^{m\times n}$ that can only be solved by $(s,t)$ linear data structures with $t\geq \Omega\left(\min\left(n,\frac{n\log{q}}{\log{s}}\right)\right)$.
\item For any $s\geq n^{1+\eps}$, every linear problem $M\in\F^{m\times n}$ can be solved by an $(s,t)$ linear data structure with $t\leq O\left(\min\left(n,\frac{n\log{q}}{\log{s}}\right)\right)$.
\end{itemize}
\end{lemma}
\begin{proof}
\leavevmode
\begin{itemize}
\item
The total number of homogeneous linear functions of $n$ arguments is $q^n$. There are $q^{sn}$ ways to choose $s$ linear functions computed in the memory cells. For a fixed choice of $s$ functions, there are at most $s^t\cdot q^{t}$ different functions which can be computed as linear compositions of $t$ out of $s$ elements. Thus, there are at most 
\begin{align*}
q^{sn} \cdot (s^t\cdot q^{t})^{m}
\end{align*}
$m$-tuples of linear functions which can be computed by data structures with $s$ memory cells. On the other hand, there are $q^{nm}$ distinct $m$-tuples of linear functions. Therefore, as long as 
\begin{align*}
q^{sn+mt} s^{mt} < q^{nm} \, ,
\end{align*}
there is a linear data structure problem with $m$ outputs which cannot be computed by a data structure with $s$ memory cells and query time $t$. Let us take $t=\frac{\eps n\log{q}}{2\log(qs)}$. Then, from $s\leq(1-\eps)m$, we have
\begin{align*}
q^{sn+mt} s^{mt}\leq q^{(1-\eps)nm}(qs)^{mt}=q^{(1-\eps)nm}(qs)^{\frac{\eps nm\log{q}}{2\log(qs)}}=q^{(1-\eps)nm}q^{\eps nm/2}<q^{nm} \,.
\end{align*}
In particular, there is a linear data structure problem which requires 
\begin{align*}
t>\frac{\eps n\log{q}}{2\log(qs)}\geq\Omega\left(\min\left(n,\frac{n\log{q}}{\log{s}}\right)\right)\,.
\end{align*}
\item
Let $\mu=1+\frac{1}{\eps}$.
It is trivial to see that any data structure can be solved with space $s=n$ and query time $t=n$. Thus, it suffices to show that if $\log{s}>2\mu\log{q}$, then $M$ can be solved by a data structure with $t\leq O\left(\frac{n\log{q}}{\log{s}}\right)$. Let 
\begin{align}\label{eq:t}
t=\left\lceil\frac{n}{\frac{\log{s}}{\mu\log{q}}-1}\right\rceil \,.
\end{align}
Let us partition the $n$ inputs into $t$ parts each of size $\floor{n/t}$ or $\ceil{n/t}$. Let the preprocessing function $P\colon\F^n\to\F^s$ compute all $q^{\ceil{n/t}}$ homogeneous linear combinations of the inputs of each part. In order to store this, we need $s'=t\cdot q^{\ceil{n/t}}$ memory cells. Now, every query can be computed as a sum of at most $t$ memory cells. It remains to show that $s'\leq s$:
\begin{align*}
s'=t\cdot q^{\ceil{n/t}}\leq n \cdot  q^{\ceil{n/t}}\leq n \cdot q^{\frac{\log{s}}{\mu\log{q}}}=n\cdot s^{\frac{1}{\mu}}\leq s^{\frac{1}{1+\eps}+\frac{1}{\mu}}\leq s\, ,
\end{align*}
where the first inequality follows from~\eqref{eq:t} and $\log{s}>2\mu\log{q}$, the second inequality is due to~\eqref{eq:t},  the third inequality follows from $s\geq n^{1+\eps}$, and the last one is due to the choice of $\mu=1+\frac{1}{\eps}$.
\end{itemize}
\end{proof}

Now we give an alternative (more constructive) proof of Lemma~\ref{lem:inner_rigidity} (with a small loss in the upper bound on the inner dimension).
\begin{lemma}\label{lem:simplelemma}
Let $M\in\F^{m\times n}$ be a matrix of rank $n$, and $V\subseteq V^m$ be its column space. If $M$ is $(m,n,r,t)$-row rigid, then $d_V(t)<n-2r$.
\end{lemma}
\begin{proof}
Assume towards a contradiction that $M$ is not $(m, n, r, t)$-row rigid. Then, by Definition~\ref{def:rigidity}, 
there exist matrices $A, B \in \F^{m\times n}$, where $A$ is $t$-sparse, and $\rk(B)\leq r$, such that $M=A+B$. 

Let $L\in\F^{m\times m}$ be a matrix representing a
linear map which vanishes on $V$: $\ker(L)=V$. 
To construct such an $L$, one can take a basis $(v_1,\ldots,v_k)$ of $V$, and extend it to a basis 
$(v_1,\ldots,v_k,w_1,\ldots,w_{m-k})$ of $F^m$. Then define $L(v_i)=0$ for $1\leq i\leq k$, and $L(w_j)=w_j$ for $1\leq j\leq m-k$, and extend $L$ by linearity. 

Now, observe that if we apply $L$ to the equality $M=A+B$, we have
$0=LM=LA+LB$, and, in particular, $\rk(LA)=\rk(LB)$. Note that the rank of the matrix on the right side is $\rk(LB)\leq\rk(B)\leq r$. By subadditivity of rank, we have $\rk(A)\geq \rk(M)-\rk(B)\geq n-r$. 

Let $U$ be the column space of $t$-sparse matrix $A$. From $\rk(A)\geq n-r$ and $\rk(LA)\leq r$, we have that 
$\dim(U \cap \ker(L))=\dim(U \cap V) \geq n-2r$, which finishes the proof.
\end{proof}

\end{document}